\documentclass[final,5p,twocolumn]{elsarticle}

\usepackage{graphicx}      
\usepackage{natbib}        

\usepackage{amsmath,amssymb,amsfonts,mathtools}
\usepackage{bm}
\usepackage{algorithm}
\usepackage{algcompatible}
\usepackage{mathrsfs}
\usepackage{nicefrac}
\usepackage{url} 

\newcommand{\mat}[1]{\ensuremath{\begin{bmatrix} #1 \end{bmatrix}}}
\newcommand{\vc}[1]{\ensuremath{\begin{bmatrix} #1 \end{bmatrix}}}
\newcommand{\ol}[1]{\overline{#1}}
\newcommand{\ul}[1]{\underline{#1}}
\newcommand{\eps}{\varepsilon}
\renewcommand{\Pr}[1]{\mathrm{Pr}\left[#1\right]}
\newcommand{\PrBig}[1]{\mathrm{Pr}\big[\,#1\;\big]}
\newcommand{\E}[1]{\mathrm{E}\left[#1\right]}
\newcommand{\norm}[1]{\left\lVert #1 \right\rVert}
\newcommand{\norminf}[1]{\left\lVert #1 \right\rVert_{\infty}}
\newcommand{\diag}[1]{\mathrm{diag}( #1 )}
\renewcommand{\exp}[1]{\textnormal{exp}\big( #1 \big)}
\renewcommand{\text}[1]{\textnormal{#1}}
\newcommand{\col}[1]{\mathrm{col}(#1)}
\newcommand{\lag}[1]{\mathrm{lag}(#1)}
\newcommand{\hankel}[2]{\bm{H}_{#1}(#2)}

\newtheorem{theorem}{Theorem}
\newtheorem{lemma}{Lemma}
\newtheorem{proposition}{Proposition}
\newtheorem{remark}{Remark}
\newtheorem{definition}{Definition}
\newtheorem{assumption}{Assumption}
\newproof{proof}{Proof}

\journal{IFAC Journal of Systems and Control}

\begin{document}

\begin{frontmatter}

\title{Stochastic Data-driven Predictive Control of Linear Systems with Sub-Gaussian Disturbances using Causal Predictors\tnoteref{pubnote} 
}
\tnotetext[pubnote]{This manuscript is a revised version of the article published in IFAC Journal of Systems and Control, vol. 25, pp. 100399, 2026 (DOI: \url{https://doi.org/10.1016/j.ifacsc.2026.100399}). It contains corrections of typos relative to the published version.}

\author{Johannes Teutsch} 
\author{Marion Leibold}

\affiliation{
            organization={Chair of Automatic Control Engineering, Department of Computer Engineering, Technical University of Munich},
            addressline={Theresienstr.~90}, 
            city={Munich},
            postcode={D-80333},
            country={Germany. Email: \{johannes.teutsch, marion.leibold\}@tum.de}}

\begin{abstract}
    We present a stochastic data-driven predictive control (DPC) framework for discrete-time linear time-invariant systems subject to sub-Gaussian additive disturbances based solely on input--output data. In contrast to related methods that rely on exact disturbance data or at least sample generation for closed-loop guarantees, the proposed approach leverages a disturbance data estimate. By enforcing consistency of the disturbance data estimate with the available input--output data and system class, we first identify data-driven and provably causal subspace predictors for use in DPC. Then, we analyze statistical properties of the corresponding prediction error, yielding tightened constraints for the nominal predictions that guarantee satisfaction of chance constraints. The proposed DPC scheme comes with guarantees on recursive feasibility and conditional chance constraint satisfaction in closed-loop under standard assumptions. A numerical evaluation study demonstrates the performance of the proposed controller.
\end{abstract}


\end{frontmatter}


\section{Introduction} \label{sec:intro}

In recent years, data-driven predictive control (DPC) has become popular as a safe and efficient control strategy for uncertain systems, utilizing measurement data without requiring an explicit model~\citep{coulson2019data,berberich2020data}. At each time step, DPC solves a finite-horizon optimal control problem (OCP) and then applies only the first input of the optimal control sequence. Unlike model predictive control (MPC), DPC leverages persistently exciting (PE) input--output data for predictions, enabled by the fundamental lemma from \cite{willems2005note}. Such data-based predictions are made either directly by linearly combining the PE data trajectories via the fundamental lemma, or indirectly by first identifying subspace predictors as in Subspace Predictive Control (SPC) \citep{favoreel1999spc}; we refer to \citep{fiedler2021relationship,dorfler2021bridging} for a discussion on direct and indirect formulations.

Crucially, the aforementioned data-based predictions are only exact for deterministic linear time-invariant (LTI) systems and noise-free input--output data. Lately, many attempts on noise-tolerant DPC have been presented \citep{berberich2020data,coulson2021distributionally,breschi2023data,chiuso2025harnessing}, but only a few come with closed-loop guarantees such as recursive feasibility and constraint satisfaction.
In particular, when additive disturbances are present, stochastic DPC as proposed by \cite{pan2021stochastic,kerz2023datadriven,pan2025towards} use distributional information to enforce probabilistic chance constraints.

Existing stochastic DPC schemes that offer closed-loop guarantees typically rely on PE input--disturbance--state/ output data to accurately capture the dynamics of disturbed LTI systems~\citep{kerz2023datadriven,pan2021stochastic,yin2023stochastic,wang2022data}. If disturbance data are not available, estimation techniques can be employed \citep{pan2021stochastic}, at the cost of closed-loop guarantees if the estimation is not exact. Although input--output data are commonly available, expecting measurable or precisely estimated disturbances is often unrealistic.

\cite{teutsch2024sampling} tackle this issue by presenting a sampling-based stochastic DPC approach that uses distributional knowledge of the disturbance to generate disturbance data samples that are consistent with the available input--output data and system class. While this allows for closed-loop guarantees, the approach is restricted to bounded disturbances and relies on a restrictive robust first-step constraint in the DPC formulation. 

Recently, \cite{ao2025stochastic} presented an MPC framework considering disturbances that follow a sub-Gaussian distribution, thus including a broad class of practical cases such as all bounded or normal distributions. Closed-loop guarantees are enabled by leveraging an indirect feedback formulation \citep{hewing2020recursively}. However, \cite{ao2025stochastic} rely on an exact system model, which can hardly be identified from noisy input--output data.

In this work, we overcome this issue by presenting a stochastic DPC scheme for LTI systems in AutoRegressive with eXogeneous input (ARX) form subject to sub-Gaussian disturbances based on input--output data. For predictions, we leverage the input--output data and a disturbance data estimate to construct subspace predictors. Under a consistency constraint on the disturbance data estimate, we show that the identified predictors obey causality --- a property known to improve the performance of DPC schemes \citep{sader2025causality} --- and allow to unify DPC and model-based ARX-MPC frameworks. In contrast to \cite{pan2021stochastic}, we explicitly consider a mismatch between the true and the estimated disturbance data by deriving and leveraging prediction error bounds. The proposed DPC scheme does not rely on restrictive robust constraints as \cite{teutsch2024sampling}; Instead, we combine classical and indirect feedback formulations \citep{knaup2024recursively} for recursive feasibility and constraint satisfaction in closed-loop, conditioned on the most recent feasible measured extended-state that is constructed from past inputs and outputs. In a numerical evaluation study, we showcase the performance of the proposed controller.

\textit{Organization:}~
Section\;\ref{sec:setup} introduces the problem setup and system assumptions. In Section\;\ref{sec:main}, we discuss the connection between consistent disturbance data and causal predictors, we analyze stochastic properties of the prediction error, and we present the proposed stochastic DPC scheme. Section\;\ref{sec:eval} presents a numerical evaluation of the proposed controller. In Section\;\ref{sec:discuss}, we discuss the presented results, before we conclude the work in Section\;\ref{sec:conclusion}.

\textit{Notation:}~
We write $\bm{I}_n$ for the identity matrix of order $n$ and $\bm{0}_{n\times m}$ for a zero matrix of dimension $n \times m$. The subscript is omitted when the dimension is clear from the context. We abbreviate the set of integers $\left\{a,\,\ldots,\, b\right\}$ by $\mathbb{N}_a^b$. The pseudo-rightinverse of a matrix $\bm{S}$ is defined as $\bm{S}^{\dagger} \coloneqq \bm{S}^{\top}\left(\bm{S}\bm{S}^{\top}\right)^{-1}$. The probability measure and the expectation operator are denoted as $\Pr{\cdot}$ and $\E{\cdot}$, respectively. With $\left[\bm{S}\right]_{[i]}$, we denote the $i$-th row/element of $\bm{S}$. The Kronecker product is denoted by ``$\otimes$". By $\col{\bm{s}_a,\,\ldots,\,\bm{s}_b} \coloneqq \mat{\bm{s}^{\top}_a\,\cdots\,\bm{s}^{\top}_b}^{\top}$, we denote the result from stacking the terms $\bm{s}_a,\,\ldots,\,\bm{s}_b$. For any sequence of vectors $\mathcal{S}_T = \left\{\bm{s}_i\right\}_{i=1}^{T}$, the Hankel matrix $\hankel{L}{\mathcal{S}_T}$ of order $L \le T$ consists of columns $\col{\bm{s}_{1+i},\ldots,\bm{s}_{L+i}}$, $i\in\mathbb{N}_0^{T-L}$. For a positive definite matrix $\bm{S} \succ \bm{0}$, we define the weighted 2-norm of the vector $\bm{s}$ as $\norm{\bm{s}}_{\bm{S}} \coloneqq \sqrt{\bm{s}^{\top} \bm{S} \bm{s}}$. The 2-norm of a vector and spectral norm of a matrix are denoted by $\norm{\cdot}$. 


\section{Problem Setup} \label{sec:setup}
In this section, we introduce the problem setup consisting of the considered system class and relevant assumptions.

We consider a disturbed LTI system $\mathscr{S}$ in ARX form, i.e.,
\begin{equation}
	\bm{y}_{k} = \bm{\Phi} \bm{\xi}_{k} + \bm{\Psi} \bm{u}_{k} + \bm{d}_{k}, \label{eq:system}
\end{equation}
where $\bm{\Phi}$, $\bm{\Psi}$ are unknown system matrices, with output $\bm{y}_{k}\in \mathbb{R}^{n_y}$, input $\bm{u}_{k} \in \mathbb{R}^{n_u}$, disturbance $\bm{d}_{k} \in \mathbb{R}^{n_y}$, and \textit{extended-state} from past $T_{\mathrm{p}} \in \mathbb{N}$ inputs and outputs, i.e.,
\begin{equation}
    \bm{\xi}_{k} \coloneqq \col{\bm{u}_{k-T_{\mathrm{p}}},\ldots,\bm{u}_{k-1},\,\bm{y}_{k-T_{\mathrm{p}}},\ldots,\bm{y}_{k-1}} \in \mathbb{R}^{n_{\xi}}, \label{eq:extstate}
\end{equation}
with $n_{\xi}\coloneqq(n_u+n_y)T_{\mathrm{p}}$.
Throughout the paper, we rely on the following assumption regarding a minimal state-space representation of system \eqref{eq:system}, common in ARX-based stochastic DPC \citep{pan2023data,teutsch2024sampling}.
\begin{assumption}\label{asm:minss}
    A minimal state-space representation 
    \begin{subequations} \label{eq:system_minimal}
    \begin{align}
    	\bm{x}_{k+1} &= \tilde{\bm{A}} \bm{x}_{k} + \tilde{\bm{B}} \bm{u}_{k} + \tilde{\bm{E}}\bm{d}_{k}, \\
    	\bm{y}_k &= \tilde{\bm{C}} \bm{x}_{k} + \tilde{\bm{D}} \bm{u}_{k} + \bm{d}_{k}
    \end{align}
    \end{subequations}
    with state $\bm{x}\in\mathbb{R}^{n_x}$, controllable $(\tilde{\bm{A}},\,[\tilde{\bm{B}} ~\tilde{\bm{E}}])$, and observable $(\tilde{\bm{A}},\,\tilde{\bm{C}})$ exists, such that for some initial state $\bm{x}_{0}$, the input--output trajectories of \eqref{eq:system} and \eqref{eq:system_minimal} coincide for all disturbance sequences $\bm{d}_k,~ k \ge 0$. 
    \hfill \small{$\square$}
\end{assumption}
Asm.~\ref{asm:minss} implies that the ARX order $T_{\mathrm{p}}$ satisfies $T_{\mathrm{p}} \ge \lag{\tilde{\bm{A}},\tilde{\bm{C}}}$, with $\lag{\tilde{\bm{A}},\tilde{\bm{C}}}$ being the smallest natural number $j \le n_x$ for which the rank of the observability matrix $\col{\tilde{\bm{C}}, \tilde{\bm{C}}\tilde{\bm{A}}, \ldots, \tilde{\bm{C}} \tilde{\bm{A}}^{j-1}}$ equals the system order $n_x$. Under this observability assumption, the equivalence of the ARX system~\eqref{eq:system} and the minimal state-space system~\eqref{eq:system_minimal} is given by \cite[Lem.~2~and~3]{sadamoto2022equivalence}. For a discussion that is tailored to the presented ARX setup with disturbance $\bm{d}_k$, see \cite[Sec.~3.1]{ou2025stochastic}. Furthermore, Asm.\;\ref{asm:minss} allows for the construction of a stabilizable and detectable state-space representation of \eqref{eq:system}, i.e.,
\begin{subequations}\label{eq:system_arx_nonminimal}
\begin{align}
    \bm{\xi}_{k+1} &= \bm{A} \bm{\xi}_{k} + \bm{B} \bm{u}_{k} + \bm{E}_y^{\top}\bm{d}_{k},\\
    \bm{y}_{k} &= \bm{\Phi} \bm{\xi}_{k} + \bm{\Psi} \bm{u}_{k} + \bm{d}_{k}, 
\end{align}
\end{subequations}
with 
$\bm{A} \coloneqq \col{\bar{\bm{A}},\,\bm{\Phi}}$, $\bm{B} \coloneqq \col{\bar{\bm{B}},\,\bm{\Psi}}$, $\bm{E}_y \coloneqq \mat{\bm{0}&\bm{I}_{n_y}}$,
\begin{equation}
    \bar{\bm{A}} \coloneqq \mat{\bm{0} & \bm{I}_{(T_{\mathrm{p}}-1)n_u} & \bm{0} & \bm{0} \\ \bm{0} & \bm{0} & \bm{0} & \bm{0} \\ \bm{0} & \bm{0} & \bm{0} & \bm{I}_{(T_{\mathrm{p}}-1)n_y} },\, \bar{\bm{B}} \coloneqq \mat{\bm{0}\\ \bm{I}_{n_u} \\ \bm{0}}. \label{eq:system_arx_nonminimal_helper}
\end{equation}
This equivalent representation of \eqref{eq:system} enables us to apply state-space methods for the system analysis \citep{bongard2022robust}, which will be useful throughout the work.

The additive disturbance $\bm{d}_k \in\mathbb{R}^{n_y}$ acting on system~\eqref{eq:system} is assumed to follow a sub-Gaussian distribution, as specified by the following assumption.
\begin{assumption} \label{asm:disturb}
    The disturbances $\bm{d}_k$, $k \ge 0$, are the realizations of a zero-mean random variable $\mathbf{d}$ that is independent and identically distributed (iid) and sub-Gaussian with variance proxy $\bm{\Sigma}_d = \bm{\Sigma}_d^{\top} \succ \bm{0}$, $\ol{\sigma}_d^{2} \coloneqq \norm{\bm{\Sigma}_d} < \infty$, i.e.,
    \begin{equation}
        \E{\exp{\bm{\lambda}^{\top}\mathbf{d}}} \le \exp{0.5\norm{\bm{\lambda}}_{\bm{\Sigma}_d}^2}
    \end{equation}
    for all $\bm{\lambda}\in\mathbb{R}^{n_y}$ \citep{ao2025stochastic}.
    \hfill \small{$\square$}
\end{assumption}
The class of sub-Gaussian distributions in Asm.\;\ref{asm:disturb} includes a wide range of probability distributions, e.g., all bounded and Gaussian distributions. Specifically for bounded disturbances such that $\bm{d}\in\mathbb{D}$ with compact $\mathbb{D}$, the variance proxy can be chosen as $\bm{\Sigma}_d = \overline{\sigma}^2 \bm{I}_{n_y}$, with $\overline{\sigma} = \max_{\bm{d}\in\mathbb{D}}\norm{\bm{d}}$. If the entries of $\mathbf{d}$ are independent and bounded by $|[\mathbf{d}]_i| \le d_{i,\max}$, then $\bm{\Sigma}_d = \diag{d^2_{1,\max},\,\ldots,\,d^2_{{n_y},\max}}$ \citep{wainwright2019high}.

Due to potentially unbounded disturbances, it is difficult to impose hard constraints on system variables. Thus, system \eqref{eq:system} is subject to (conditioned) chance constraints for outputs and inputs. Specifically, given compact polytopic sets $\mathbb{Y} = \left\{\bm{y} \in \mathbb{R}^{n_y} \left| \bm{G}_y \,\bm{y} \le \bm{g}_y \hspace{-2pt}\right.\right\}$, $\mathbb{U} = \left\{\bm{u} \in \mathbb{R}^{n_u} \left| \bm{G}_u \bm{u} \le \bm{g}_u \hspace{-2pt}\right.\right\}$
containing the origin, with constraint parameters $\bm{G}_y \in \mathbb{R}^{r_y \times n_y}$, $\bm{g}_y \in \mathbb{R}^{r_y}$, $\bm{G}_u \in \mathbb{R}^{r_u \times n_u}$, $\bm{g}_u \in \mathbb{R}^{r_u}$, the individual chance constraints are defined as
\begin{subequations}\label{eq:constraints}
	\begin{align}
	\Pr{\left[\bm{G}_y\right]_i \bm{y}_{k} \le \left[\bm{g}_y\right]_i~\big|~\bm{\xi}_{\tilde{k}}}\, &\ge 1-\eps^y_{i} ~~~ \forall i\in\mathbb{N}_1^{r_y}, \label{eq:outputcons} \\
	\Pr{\left[\bm{G}_u\right]_j \bm{u}_{k} \le \left[\bm{g}_u\right]_j~\big|~\bm{\xi}_{\tilde{k}}} &\ge 1-\eps^u_{j} ~~~ \forall j\in\mathbb{N}_1^{r_u}, \label{eq:inputcons}
	\end{align}
\end{subequations}
with $\tilde{k} \in \mathbb{N}_0^k$ and the risk parameters $\eps^y_{i}$, $\eps^u_{j} \in (0,\,1]$. 
Notably, the probabilities in~\eqref{eq:constraints} are conditioned on the measurement $\bm{\xi}_{\tilde{k}}$. Thus, the chance constraints~\eqref{eq:constraints} are evaluated using the last $k-\tilde{k}+1$ random disturbances $\bm{d}_{\tilde{k}},\dots,\bm{d}_{k}$, treating the first $\tilde{k}$ realizations as fixed.
In stochastic MPC literature, a common choice of $\tilde{k}$ is $\tilde{k} = 0$ for indirect feedback approaches and $\tilde{k} = k$ for direct feedback approaches \citep{hewing2020recursively}.

In stochastic optimal control, the control input $\bm{u}_k$ is typically parameterized by some feedback law for disturbance attenuation over the prediction horizon. Particularly in this work, we consider affine control policies of the form
\begin{equation} \label{eq:inputdecomp}
    \bm{u}_k = \bm{K} \bm{\xi}_k + \bm{v}_k,
\end{equation}
with a fixed feedback gain $\bm{K}$ and free correction input $\bm{v}_k$. Hence, considering the system setup introduced thus far, we aim to solve the (conceputal) stochastic OCP
\begin{subequations} \label{eq:ocp_original}
		\begin{align}
		& \underset{\bm{v}_k,~k \,\ge\, 0}{\mathrm{minimize}} ~~~ \mathrm{E}\Bigg[\,\sum\limits_{k=0}^{\infty} \left( \norm{\bm{y}_{k}}^2_{\bm{Q}} + \norm{\bm{u}_{k}}^2_{\bm{R}} \right)\Bigg] &\label{eq:ocp_original_cost}\\		
		\mathrm{s.t.}~~ &  \eqref{eq:system},~  \eqref{eq:constraints},~\eqref{eq:inputdecomp},~\text{Asm.\;\ref{asm:disturb}} ~~\forall k\ge0,
		\end{align}
\end{subequations}
with initial condition $\bm{\xi}_0$ and weighting matrices $\bm{Q}$, $\bm{R} \succ \bm{0}$. However, the OCP \eqref{eq:ocp_original} is intractable since the system parameters $\bm{\Phi}$ and $\bm{\Psi}$ and disturbances $\bm{d}_k$ in \eqref{eq:system} are unknown in our problem setting, and due to the probabilistic formulation of the constraints~\eqref{eq:constraints}. 
Therefore, the goal of the upcoming section is to derive a DPC scheme based on a tractable receding horizon approximation of the OCP~\eqref{eq:ocp_original} that is recursively feasible and guarantees satisfaction of the chance constraints~\eqref{eq:constraints} in closed-loop.
Given that the system parameters $\bm{\Phi}$ and $\bm{\Psi}$ are unknown, we base predictions in the proposed DPC scheme on PE input–output data, specified as follows \citep{teutsch2024sampling}.
\begin{definition}\label{def:persistency}
	A trajectory $\mathcal{S}_T = \left\{\bm{s}_{i}\right\}_{i=1}^{T}$ of length $T\in\mathbb{N}$ with $\bm{s}_i \in \mathbb{R}^{n_s}$ is PE of order $L \le T$ if the Hankel matrix $\hankel{L}{\mathcal{S}_T}$ has full row-rank~$n_s L$.
    \hfill \small{$\square$}
\end{definition}
\begin{assumption}\label{asm:trajData}
    An input--output trajectory $\{\bm{u}^{\mathrm{d}}_i\}_{i=1-T_{\mathrm{p}}}^{T}$, $\{\bm{y}^{\mathrm{d}}_i\}_{i=1-T_{\mathrm{p}}}^{T}$ of system \eqref{eq:system} is available, yielding the data
    $\mathcal{U}_T \coloneqq \{\bm{u}^{\mathrm{d}}_i\}_{i=1}^{T}$, $\mathcal{Y}_T \coloneqq\{\bm{y}^{\mathrm{d}}_i\}_{i=1}^{T}$, $\mathcal{X}_{T+1} \coloneqq \{\bm{\xi}^{\mathrm{d}}_i\}_{i=1}^{T+1}$ via \eqref{eq:extstate}, and $\mathcal{V}_T \coloneqq \{\bm{v}^{\mathrm{d}}_i\}_{i=1}^{T}$ via \eqref{eq:inputdecomp}. The corresponding disturbance data $\mathcal{D}_T \coloneqq \{\bm{d}^{\mathrm{d}}_{i}\}_{i=1}^{T}$ is unknown, but the following holds:
    \begin{itemize}
        \item[(a)] the trajectory of generalized inputs $\{\col{\bm{v}^{\mathrm{d}}_i,\,\bm{d}^{\mathrm{d}}_i}\}_{i=1}^{T}$ is PE of order $n_x+T_{\mathrm{f}}+T_{\mathrm{p}}$ with horizon $T_{\mathrm{f}} \in \mathbb{N}$, 
        \item[(b)] the matrix $\col{\hankel{1}{\mathcal{X}_{T}},\,\hankel{1}{\mathcal{V}_{T}}}$ has full row rank. \hfill \small{$\square$}
    \end{itemize}
\end{assumption}
\begin{remark}
    Asm.\;\ref{asm:trajData} is not restrictive in practice: Since the (unknown) $\mathcal{D}_T$ is realized by a random process (see Asm.\;\ref{asm:disturb}), (a) can be satisfied by choosing appropriate correction terms $\bm{v}^{\mathrm{d}}_i$ and applying~\eqref{eq:inputdecomp} for the data collection. Alternatively, when input data $\mathcal{U}_{T}$ are available, $\mathcal{V}_{T} \coloneqq \{\bm{v}^{\mathrm{d}}_{i}\}_{i=1}^{T}$ can be constructed using $\bm{v}^{\mathrm{d}}_{i} = \bm{u}^{\mathrm{d}}_i - \bm{K} \bm{\xi}^{\mathrm{d}}_i$ for a given $\bm{K}$. In case the matrix in (b) does not have full rank despite the perturbation by the disturbance, defining an alternative minimal state provides a remedy \citep{alsalti2023notes}. \hfill \small{$\square$}
\end{remark}

Related works on stochastic DPC use input--disturbance--output data similar to Asm.~\ref{asm:trajData} for predictions \citep{pan2023data,yin2023stochastic,wang2022data}. However, crucially, we assume that the corresponding disturbance data $\mathcal{D}_T$ is \textit{unknown} (see Asm.~\ref{asm:trajData}), introducing uncertainty into predictions.
To mitigate this issue and provide closed-loop guarantees without the availability of the true disturbance data $\mathcal{D}_T$, we make the following assumption on the system parameters in \eqref{eq:system} and on the feedback gain $\bm{K}$ from \eqref{eq:inputdecomp}.
\begin{assumption} \label{asm:strongstab}
    Consider systems~\eqref{eq:system}~and~\eqref{eq:system_arx_nonminimal} under input~\eqref{eq:inputdecomp} and data as in Asm.\;\ref{asm:trajData}. The following holds:
    \begin{itemize}
        \item[(a)] the unknown system parameters $(\bm{\Phi},\bm{\Psi})$ are bounded by a given compact polytopic set $\mathbb{A}^{\mathrm{d}}$ according to $(\bm{A},\bm{B}) \in \mathbb{A}^{\mathrm{d}}$, with $\bm{A} = \col{\bar{\bm{A}},\,\bm{\Phi}}$ and $\bm{B} = \col{\bar{\bm{B}},\,\bm{\Psi}}$ from system \eqref{eq:system_arx_nonminimal}, and $\bar{\bm{A}}$, $\bar{\bm{B}}$ from \eqref{eq:system_arx_nonminimal_helper},
        \item[(b)] the gain $\bm{K}$ is $(c_{\xi},\gamma_{\xi})$-strongly stabilizing for all $(\bm{A},\bm{B}) \in \mathbb{A}^{\mathrm{d}}$ with constants $c_{\xi}>0$ and $\gamma_{\xi} \in (0,1)$, i.e., $\norm{\bm{A}_{\mathrm{cl}}^k} \le c_{\xi}\gamma_{\xi}^k$ for all $k\ge 0$, with $\bm{A}_{\mathrm{cl}} = \bm{A} + \bm{B}\bm{K}$,
        \item[(c)] a constant $\gamma_v \in (0,\infty)$ that satisfies
        \begin{equation}
            \max_{\bm{\xi} \in \mathbb{X},~\bm{v} + \bm{K}\bm{\xi} \in \mathbb{U}} \norm{\mat{\hankel{1}{\mathcal{X}_{T}}\\\hankel{1}{\mathcal{V}_{T}}}^{\dagger}\vc{\bm{\xi}\\\bm{v}}} \le \gamma_v\label{eq:inputstatebound}
        \end{equation}
        is known, with $\mathbb{X}$ being the constraint set for the extended-state~\eqref{eq:extstate} corresponding to the input and output constraint sets $\mathbb{U}$ and $\mathbb{Y}$. \hfill \small{$\square$} 
    \end{itemize}
\end{assumption}
\begin{remark} \label{rem:strongstab}
    Assuming bounds on the system parameters as in (a) is common in robust and stochastic predictive control \citep{lorenzen2019robust,arcari2023stochastic,teutsch2024adaptive}. Such parameter bounds can be obtained using bounds on the disturbance data $\mathcal{D}_T$ \citep{berberich2020robust,alanwar2023data,van2023informativity,teutsch2024sampling}. In case that a compact set $\mathbb{A}^{\mathrm{d}}$ for (a) cannot be obtained (e.g., due to unbounded disturbances), $\mathbb{A}^{\mathrm{d}}$ can be replaced by a confidence set \citep{van2023quadratic}. Furthermore, strong stability \citep{cohen2018online,kerz2024safe} in (b) is related to quadratic stability of linear systems (i.e., existence of a common Lyapunov function over a bounded set of system models), which is also a common assumption in robust and stochastic predictive control \citep{lorenzen2019robust,arcari2023stochastic,teutsch2024sampling}. Essentially, the notion of strong stability used here quantifies decay rates of quadratic stability via the constants $c_{\xi}$ and $\gamma_{\xi}$. Just as for quadratic stability, a suitable gain $\bm{K}$ and corresponding constants $c_{\xi}$ and $\gamma_{\xi}$ can be derived with the help of linear matrix inequalities involving the vertices of the model set $\mathbb{A}^{\mathrm{d}}$, as detailed in \ref{app:strongstab}.
    Lastly, the constant $\gamma_v$ can readily be computed via the left-hand side of \eqref{eq:inputstatebound} using the vertices of the set $\{\bm{\xi},\bm{v} \,|\,\bm{\xi} \in \mathbb{X},~\bm{v} + \bm{K}\bm{\xi} \in \mathbb{U}\}$. \hfill \small{$\square$}
\end{remark}


\section{Method} \label{sec:main}
In this section, we will derive the proposed stochastic DPC scheme for system~\eqref{eq:system}. Towards this goal, we first revisit data-driven system representations that we use as the basis for predictions from data (see Asm.~\ref{asm:trajData}). In Section~\ref{sec:consistency}, we discuss the concept of consistent disturbance data $\mathcal{D}_T$ from \cite{teutsch2024sampling} and generalize it to causal multi-step predictions (Prop.~\ref{prop:consistency_multistep}). This will be used in Section~\ref{sec:predictors} to derive causal data-driven predictors based on a disturbance data estimate $\hat{\mathcal{D}}_T$, since the true disturbance data $\mathcal{D}_T$ is unknown (see Asm.~\ref{asm:trajData}). In Section~\ref{sec:prederr}, we analyze statistical properties of the corresponding prediction error (Lem.~\ref{lem:prederror_data} and \ref{lem:prederror_add}), leveraging the disturbance properties from Asm.~\ref{asm:disturb} and the system parameter bounds from Asm.~\ref{asm:strongstab}. This allows us to define tightened constraints for the proposed nominal predictions (Lem.~\ref{lem:constight}), functioning as a deterministic approximation of the chance constraints~\eqref{eq:constraints}. Lastly, we present the proposed stochastic DPC scheme and its theoretical guarantees (Thm.~\ref{thm:properties}) based on the proposed data-driven predictors and tightened constraints.

Consider a system $\mathscr{S}$ of the form \eqref{eq:system} satisfying Asm.\;\ref{asm:minss}, with input \eqref{eq:inputdecomp} and data $\mathcal{V}_T$, $\mathcal{D}_T$, $\mathcal{Y}_T$, and $\mathcal{X}_{T}$ satisfying Asm.\;\ref{asm:trajData}(a). Via \cite[Lem.~2]{teutsch2024sampling}, derived from the fundamental lemma by \cite{willems2005note} and similarly presented by \cite{pan2023data,pan2021stochastic}, any length-$\left(T_{\mathrm{p}} + T_{\mathrm{f}}\right)$ input-disturbance-output trajectory $\left\{\bm{v}_{i}\right\}_{i=k-T_{\mathrm{p}}}^{k+T_{\mathrm{f}}-1}$, $\left\{\bm{d}_{i}\right\}_{i=k}^{k+T_{\mathrm{f}}-1}$, $\left\{\bm{y}_{i}\right\}_{i=k-T_{\mathrm{p}}}^{k+T_{\mathrm{f}}-1}$ is a valid trajectory of $\mathscr{S}$ for $k \ge 0$ if and only if there exists $\bm{\alpha} \in \mathbb{R}^{T-T_{\mathrm{f}}+1}$ such that
\begin{equation} \label{eq:extfundlemm}
    \vc{\bm{\xi}_k \\ \col{\bm{v}_{k},\,\ldots,\,\bm{v}_{k+T_{\mathrm{f}}-1}}\\ \col{\bm{d}_{k},\,\ldots,\,\bm{d}_{k+T_{\mathrm{f}}-1}}\\ \col{\bm{y}_{k},\,\ldots,\,\bm{y}_{k+T_{\mathrm{f}}-1}}} = \mat{\hankel{1}{\mathcal{X}_{T-T_{\mathrm{f}}+1}} \\ \hankel{T_{\mathrm{f}}}{\mathcal{V}_{T}} \\ \hankel{T_{\mathrm{f}}}{\mathcal{D}_{T}} \\ \hankel{T_{\mathrm{f}}}{\mathcal{Y}_{T}}} \bm{\alpha},
\end{equation}
with $\bm{\xi}_k$ and $\mathcal{X}_{T-T_{\mathrm{f}}+1} = \{\bm{\xi}^{\mathrm{d}}_i\}_{i=1}^{T-T_{\mathrm{f}}+1}$ according to \eqref{eq:extstate}.
By varying $\bm{\alpha}$, the data-driven system representation \eqref{eq:extfundlemm} enables predictions when disturbance data $\mathcal{D}_{T}$ are available, which notably is not the case in this work (see Asm.\;\ref{asm:trajData}). To still use \eqref{eq:extfundlemm} for predictions in this case, two approaches have been presented in the literature: \cite{teutsch2024sampling} employ data-consistent samples of the unknown disturbance data $\mathcal{D}_{T}$, whereas \cite{pan2023data,pan2021stochastic} derive a data-consistent estimate $\hat{\mathcal{D}}_{T}$ of $\mathcal{D}_{T}$. While the former requires a sample generation mechanism, the latter suffers from the inexactness of the disturbance data estimate, resulting in a loss of closed-loop guarantees. In this work, we overcome these issues by explicitly analyzing and leveraging the prediction error that results from using an estimate $\hat{\mathcal{D}}_{T}$ instead of the unknown disturbance data $\mathcal{D}_{T}$ for predictors based on \eqref{eq:extfundlemm}. To this end, let us first revisit the concept of consistent disturbance data $\mathcal{D}_T$ from \cite{teutsch2024sampling} and extend it to multi-step predictions as via \eqref{eq:extfundlemm}.

\subsection{Consistent Disturbance Data} \label{sec:consistency}
Although the disturbance data $\mathcal{D}_{T}$ are unknown, we know that they need to be consistent with the available input--output data from Asm.\;\ref{asm:trajData} and with the system class \eqref{eq:system}. Specifically, consider the data matrices
\begin{align*}
    \hankel{1}{\mathcal{X}_{T}} = \mat{\bm{\xi}^{\mathrm{d}}_1 & \cdots & \bm{\xi}^{\mathrm{d}}_T},\, \hankel{1}{\mathcal{Y}_{T}} = \mat{\bm{y}^{\mathrm{d}}_1 & \cdots & \bm{y}^{\mathrm{d}}_T}, \\
    \hankel{1}{\mathcal{V}_{T}} = \mat{\bm{v}^{\mathrm{d}}_1 & \cdots & \bm{v}^{\mathrm{d}}_T},\, \hankel{1}{\mathcal{D}_{T}} = \mat{\bm{d}^{\mathrm{d}}_1 & \cdots & \bm{d}^{\mathrm{d}}_T}.
\end{align*}
As the available input--output data $\mathcal{V}_{T}$, $\mathcal{Y}_{T}$ and unknown disturbance data $\mathcal{D}_{T}$ stem from system \eqref{eq:system} with input parameterization~\eqref{eq:inputdecomp}, the above data matrices must obey
\begin{equation}
	\hankel{1}{\mathcal{Y}_{T}} = \bm{\Phi}_{\mathrm{cl}} \hankel{1}{\mathcal{X}_{T}} + \bm{\Psi} \hankel{1}{\mathcal{V}_{T}} + \hankel{1}{\mathcal{D}_{T}}, \label{eq:dynamics_data}
\end{equation}
with $\bm{\Phi}_{\mathrm{cl}} \coloneqq \bm{\Phi} + \bm{\Psi}\bm{K}$.
Equation \eqref{eq:dynamics_data} allows us to define a constraint on the disturbance data $\mathcal{D}_{T}$ ensuring consistency with the available input--output data $\mathcal{V}_{T}$, $\mathcal{Y}_{T}$ and the underlying system class \eqref{eq:system}, which has been widely used in the literature \citep{pan2021stochastic,berberich2020robustfb,alanwar2023data,teutsch2024sampling}. 
This is specified by the following result \cite[Prop.\;1]{teutsch2024sampling}.
\begin{proposition} \label{prop:consistency}
    Consider data $\mathcal{V}_{T}$, $\mathcal{D}_{T}$, $\mathcal{Y}_{T}$ of system \eqref{eq:system} under input~\eqref{eq:inputdecomp} satisfying Asm.\;\ref{asm:trajData}(b) with unknown $\mathcal{D}_{T}$. Then,
    \begin{equation}
        \left(\hankel{1}{\mathcal{Y}_{T}} - \hankel{1}{\mathcal{D}_{T}}\right) \bm{\Pi}^{v} = \bm{0}\label{eq:consistency}
    \end{equation}
    with $\bm{\Pi}^{v} \coloneqq \bm{I}_T - \bm{S}_1^{\dagger}\bm{S}_1$ and $\bm{S}_1 \coloneqq \col{\hankel{1}{\mathcal{X}_{T}},\hankel{1}{\mathcal{V}_{T}}}$.

    Further, any $\hat{\mathcal{D}}_{T}$ satisfying \eqref{eq:consistency} implicitly determines parameters $\hat{\bm{\Phi}}_{\mathrm{cl}}$, $\hat{\bm{\Psi}}$ of a system of form \eqref{eq:system}, i.e.,
    \begin{equation}
        \mat{\hat{\bm{\Phi}}_{\mathrm{cl}} & \hat{\bm{\Psi}}} = \left(\hankel{1}{\mathcal{Y}_{T}} - \hankel{1}{\hat{\mathcal{D}}_{T}}\right) \bm{S}_1^{\dagger}, \label{eq:connection_sysparams_data}
    \end{equation}
    such that \eqref{eq:dynamics_data} is satisfied considering the input \eqref{eq:inputdecomp}. \hfill \small{$\square$}
\end{proposition}

Prop.\;\ref{prop:consistency} allows to define disturbance data \emph{consistent} with the available input--output data and system class, namely any disturbance data $\mathcal{D}_{T}$ that satisfies~\eqref{eq:consistency}. 
Crucially, Prop.\;\ref{prop:consistency} only considers Hankel matrices of depth\;$1$, and thus produces $1$-step predictors via relation~\eqref{eq:connection_sysparams_data}. In contrast, \eqref{eq:extfundlemm} functions as a direct $T_{\mathrm{f}}$-step predictor in DPC. To investigate such $T_{\mathrm{f}}$-step predictors, we generalize Prop.\;\ref{prop:consistency} to Hankel matrices of depth $T_{\mathrm{f}} \in \mathbb{N}$ next. 

Given the extended-state-space representation~\eqref{eq:system_arx_nonminimal} with input parameterization \eqref{eq:inputdecomp}, let us define the extended observability matrix $\mathcal{O}_{T_{\mathrm{f}}} \coloneqq \col{\bm{\Phi}_{\mathrm{cl}}, \bm{\Phi}_{\mathrm{cl}} \bm{A}_{\mathrm{cl}}, \ldots, \bm{\Phi}_{\mathrm{cl}} \bm{A}_{\mathrm{cl}}^{T_{\mathrm{f}}-1}}$ and the Toeplitz matrices $\mathcal{T}_{T_{\mathrm{f}}}^{v} \coloneqq \mathcal{T}_{T_{\mathrm{f}}}(\bm{A}_{\mathrm{cl}},\,\bm{B},\,\bm{\Phi}_{\mathrm{cl}},\,\bm{\Psi})$, $\mathcal{T}_{T_{\mathrm{f}}}^{d} \coloneqq \mathcal{T}_{T_{\mathrm{f}}}(\bm{A}_{\mathrm{cl}},\,\bm{E}_y^\top,\,\bm{\Phi}_{\mathrm{cl}},\,\bm{I}_{n_y})$ with $\bm{A}_{\mathrm{cl}} \coloneqq \bm{A} + \bm{B}\bm{K}$ and
\begin{equation} \label{eq:toeplitz}
    \mathcal{T}_{T_{\mathrm{f}}}(\bm{A},\,\bm{B},\,\bm{C},\,\bm{D}) \coloneqq  \mat{\bm{D} & \bm{0} & \cdots & \bm{0} \\ \bm{C}\bm{B} & \bm{D} & \ddots & \vdots \\ \vdots & \ddots & \ddots & \bm{0} \\ \bm{C}\bm{A}^{T_{\mathrm{f}}-2}\bm{B} & \cdots & \bm{C}\bm{B} & \bm{D}}.
\end{equation}
This yields the $T_{\mathrm{f}}$-step form of the system \eqref{eq:system} as
\begin{equation} \label{eq:system_multistep}
    \bm{y}_{\mathrm{f},k} = \mathcal{O}_{T_{\mathrm{f}}} \bm{\xi}_k + \mathcal{T}_{T_{\mathrm{f}}}^{v} \bm{v}_{\mathrm{f},k}  + \mathcal{T}_{T_{\mathrm{f}}}^{d} \bm{d}_{\mathrm{f},k},
\end{equation}
with vectors of future outputs $\bm{y}_{\mathrm{f},k} \coloneqq \col{\bm{y}_{k},\,\ldots,\,\bm{y}_{k+T_{\mathrm{f}}-1}}$, future correction inputs $\bm{v}_{\mathrm{f},k} \coloneqq \col{\bm{v}_{k},\,\ldots,\,\bm{v}_{k+T_{\mathrm{f}}-1}}$, and future disturbances $\bm{d}_{\mathrm{f},k} \coloneqq \col{\bm{d}_{k},\,\ldots,\,\bm{d}_{k+T_{\mathrm{f}}-1}}$. We remark that the corresponding $T_{\mathrm{f}}$-step form of the extended-state dynamics in \eqref{eq:system_arx_nonminimal} for the vector of future extended-states $\bm{\xi}_{\mathrm{f},k} \coloneqq \col{\bm{\xi}_{k+1},\,\ldots,\,\bm{\xi}_{k+T_{\mathrm{f}}}}$ follows from \eqref{eq:system_multistep} by leveraging $\bm{\Phi}_{\mathrm{cl}} = \bm{E}_y\bm{A}_{\mathrm{cl}}$ (respectively, $\bm{y}_k = \bm{E}_y \bm{\xi}_{k+1}$) and omitting all $\bm{E}_y$ in \eqref{eq:system_multistep} that are multplied from the left. Similarly, the $T_{\mathrm{f}}$-step form of the input parameterization \eqref{eq:inputdecomp} follows by leveraging $\bm{K} = \bm{E}_u\bm{A}_{\mathrm{cl}}$ (respectively, $\bm{u}_k = \bm{E}_u \bm{\xi}_{k+1}$), with $\bm{E}_u \coloneqq \mat{\bm{0}_{n_u \times n_u(T_{\mathrm{p}}-1)}& \bm{I}_{n_u}& \bm{0}_{n_u\times n_yT_{\mathrm{p}}} }$. 
Following \cite{de2019formulas}, we can write the $T_{\mathrm{f}}$-step version of \eqref{eq:dynamics_data} (i.e., the data version of \eqref{eq:system_multistep}) as
\begin{equation} \label{eq:dynamics_data_multistep}
    \hankel{T_{\mathrm{f}}}{\mathcal{Y}_{T}} = \mathcal{O}_{T_{\mathrm{f}}} \hankel{1}{\mathcal{X}_{\tilde{T}}} + \mathcal{T}_{T_{\mathrm{f}}}^{v} \hankel{T_{\mathrm{f}}}{\mathcal{V}_{T}} + \mathcal{T}_{T_{\mathrm{f}}}^{d}\hankel{T_{\mathrm{f}}}{\mathcal{D}_{T}},
\end{equation}
with $\tilde{T} \coloneqq T-T_{\mathrm{f}}+1$.
The next result generalizes Prop.\;\ref{prop:consistency}.
\begin{proposition} \label{prop:consistency_multistep}
    Consider data $\mathcal{V}_{T}$, $\mathcal{D}_{T}$, $\mathcal{Y}_{T}$ of system \eqref{eq:system} satisfying Asm.\;\ref{asm:trajData} with unknown $\mathcal{D}_{T}$. Then,
    \begin{equation}
        \left(\hankel{T_{\mathrm{f}}}{\mathcal{Y}_{T}} - \mathcal{T}_{T_{\mathrm{f}}}^{d}\hankel{T_{\mathrm{f}}}{\mathcal{D}_{T}}\right)\bm{\Pi}^{v}_{T_{\mathrm{f}}} = \bm{0}\label{eq:consistency_multistep}
    \end{equation}
    with $\bm{\Pi}^{v}_{T_{\mathrm{f}}} \coloneqq \bm{I}_{\tilde{T}} - \bm{S}_{2}^{\dagger}\bm{S}_{2}$ and $\bm{S}_{2} \coloneqq \col{\hankel{1}{\mathcal{X}_{\tilde{T}}},\hankel{{T_{\mathrm{f}}}}{\mathcal{V}_{T}}}$.

    Further, any $\hat{\mathcal{D}}_{T}$ satisfying \eqref{eq:consistency} and Asm.\;\ref{asm:trajData} implicitly determines matrices $\hat{\mathcal{O}}_{T_{\mathrm{f}}}$, $\hat{\mathcal{T}}_{T_{\mathrm{f}}}^{v}$, $\hat{\mathcal{T}}_{T_{\mathrm{f}}}^{d}$ of a $T_{\mathrm{f}}$-step system \eqref{eq:system_multistep} such that \eqref{eq:dynamics_data_multistep} and \eqref{eq:consistency_multistep} are satisfied, with
    \begin{align}
        &\mat{\hat{\mathcal{O}}_{T_{\mathrm{f}}}\, \hat{\mathcal{T}}_{T_{\mathrm{f}}}^{v}} = \big(\hankel{T_{\mathrm{f}}}{\mathcal{Y}_{T}} - \hat{\mathcal{T}}_{T_{\mathrm{f}}}^{d}\hankel{T_{\mathrm{f}}}{\hat{\mathcal{D}}_{T}}\big) \bm{S}_{2}^{\dagger}, \label{eq:connection_sysparams_data_multistep}\\
        &\mat{\hat{\mathcal{O}}_{T_{\mathrm{f}}} & \hat{\mathcal{T}}_{T_{\mathrm{f}}}^{v} & \hat{\mathcal{T}}_{T_{\mathrm{f}}}^{d}} = \hankel{T_{\mathrm{f}}}{\mathcal{Y}_{T}} \col{\bm{S}_{2},\hankel{T_{\mathrm{f}}}{\hat{\mathcal{D}}_{T}}}^{\dagger},\label{eq:connection_sysparams_data_multistep_2}
    \end{align}
     where \eqref{eq:connection_sysparams_data_multistep} and \eqref{eq:connection_sysparams_data_multistep_2} are equivalent.
    \hfill \small{$\square$}
\end{proposition}
\begin{proof}
    See \ref{app:proof_consistency}.
    \hfill \small{$\blacksquare$}
\end{proof}

Evidently, for a candidate $\hat{\mathcal{D}}_T$, direct verification of the $T_{\mathrm{f}}$-step consistency constraint~\eqref{eq:consistency_multistep} is hindered by the unknown entries of the corresponding Toeplitz matrix~$\hat{\mathcal{T}}_{T_{\mathrm{f}}}^{d}$. However, via Prop.\;\ref{prop:consistency_multistep}, satisfaction of the $1$-step consistency constraint~\eqref{eq:consistency} in conjunction with the PE conditions from Asm.\;\ref{asm:trajData} implies $T_{\mathrm{f}}$-step consistency \eqref{eq:consistency_multistep} of $\hat{\mathcal{D}}_T$. Note that this fact has implicitly been used in our previous work \citep{teutsch2024sampling}, but not formally proven.

When leveraging a consistent estimate $\hat{\mathcal{D}}_{T}$ via Prop.\;\ref{prop:consistency_multistep} to build predictors based on \eqref{eq:extfundlemm}, the $T_{\mathrm{f}}$-step consistency constraint~\eqref{eq:consistency} enforces predictions to obey causality, expressed via the lower-triangular structure in the Toeplitz matrices in \eqref{eq:dynamics_data_multistep} and \eqref{eq:connection_sysparams_data_multistep}, \eqref{eq:connection_sysparams_data_multistep_2}. Evidently, the Toeplitz structure also renders the corresponding predictors inherently time-invariant. In the following section, we will derive such causal and time-invariant predictors.

\subsection{Causal Predictors} \label{sec:predictors}
We now identify output predictors based on \eqref{eq:extfundlemm} and Prop.\;\ref{prop:consistency_multistep}. Consider the vectors of future correction inputs $\bm{v}_{\mathrm{f},k}$, disturbances $\bm{d}_{\mathrm{f},k}$, outputs $\bm{y}_{\mathrm{f},k}$, and extended-states $\bm{\xi}_{\mathrm{f},k}$ from \eqref{eq:system_multistep}, and consider the true (unknown) disturbance data $\mathcal{D}_T$. As discussed in our previous work \citep{teutsch2024sampling}, there exists an $\bm{\alpha}$ that satisfies~\eqref{eq:extfundlemm} with fixed initial condition $\bm{\xi}_k$ and sequence of correction inputs $\bm{v}_{\mathrm{f},k}$.
More specifically, all solutions for $\bm{\alpha}$ can be expressed as
\begin{equation} \label{eq:pred_alpha}
    \bm{\alpha} = {\bm{M}}({\mathcal{D}}_{T})^{\dagger} \col{\bm{\xi}_k, \bm{v}_{\mathrm{f},k}, \bm{d}_{\mathrm{f},k}} + \bm{\Pi}^{\alpha}\left(\mathcal{D}_{T}\right) \tilde{\bm{\alpha}},
\end{equation}
with ${\bm{M}}({\mathcal{D}}_{T}) \coloneqq \col{\hankel{1}{{\mathcal{X}}_{\tilde{T}}}, \hankel{T_{\mathrm{f}}}{\mathcal{V}_{T}}, \hankel{T_{\mathrm{f}}}{{\mathcal{D}}_{T}}}$, some matrix $\bm{\Pi}^{\alpha}\left(\mathcal{D}_{T}\right)$ whose columns span the null space of the data matrix ${\bm{M}}({\mathcal{D}}_{T})$, and free variable $\tilde{\bm{\alpha}} \in \mathbb{R}^{\tilde{T}-n_{\xi}-(n_u+n_y)T_{\mathrm{f}}}$. By applying $\bm{\alpha}$ from \eqref{eq:pred_alpha} to \eqref{eq:extfundlemm}, we obtain the exact data-driven subspace predictor 
\begin{equation}
    \bm{y}_{\mathrm{f},k} = \hankel{T_{\mathrm{f}}}{\mathcal{Y}_{T}}{\bm{M}}({\mathcal{D}}_{T})^{\dagger} \col{\bm{\xi}_k , \bm{v}_{\mathrm{f},k}, \bm{d}_{\mathrm{f},k}}. \label{eq:predictor_output_exact}
\end{equation}
Note that via \eqref{eq:connection_sysparams_data_multistep_2} from Prop.\;\ref{prop:consistency_multistep}, the predictor \eqref{eq:predictor_output_exact} is equivalent to the model-based predictor~\eqref{eq:system_multistep} due to consistency of the true (unknown) disturbance data $\mathcal{D}_{T}$. Furthermore, \eqref{eq:predictor_output_exact} is independent of the free variable~$\tilde{\bm{\alpha}}$ since, via Prop.\;\ref{prop:consistency_multistep}, we can leverage \eqref{eq:dynamics_data_multistep} to obtain
\begin{align*}
    \hankel{T_{\mathrm{f}}}{\mathcal{Y}_{T}}\bm{\Pi}^{\alpha}\tilde{\bm{\alpha}} = \mat{\mathcal{O}_{T_{\mathrm{f}}} & \mathcal{T}_{T_{\mathrm{f}}}^{v} & \mathcal{T}_{T_{\mathrm{f}}}^{d}}{\bm{M}}({\mathcal{D}}_{T})\bm{\Pi}^{\alpha}\tilde{\bm{\alpha}} = \bm{0}
\end{align*}
by definition of the matrix $\bm{\Pi}^{\alpha}$. 
The corresponding exact extended-state subspace predictor is given as
\begin{equation}
    \bm{\xi}_{\mathrm{f},k} =  \bm{H}_{\mathrm{f},\xi}{\bm{M}}({\mathcal{D}}_{T})^{\dagger} \col{\bm{\xi}_k, \bm{v}_{\mathrm{f},k}, \bm{d}_{\mathrm{f},k}},\label{eq:predictor_extstate_exact}
\end{equation}
with $\bm{H}_{\mathrm{f},\xi} \coloneqq \mat{\mathbf{0}_{n_{\xi}T_{\mathrm{f}}\times n_{\xi}}&\boldsymbol{I}_{n_{\xi}T_{\mathrm{f}}}}\hankel{T_{\mathrm{f}}+1}{\mathcal{X}_{T+1}}$. We refer to \cite{fiedler2021relationship} for a discussion on the usage of such subspace predictors in DPC.

Evidently, the exact predictors \eqref{eq:predictor_output_exact} and \eqref{eq:predictor_extstate_exact} require the true disturbance data $\mathcal{D}_T$ and future disturbances $\bm{d}_{\mathrm{f},k}$. Since neither quantity is available in the considered problem setup (Asm.~\ref{asm:trajData}), we employ estimations to construct predictors corresponding to \eqref{eq:predictor_output_exact} and \eqref{eq:predictor_extstate_exact}. By the zero-mean property from Asm.\;\ref{asm:disturb}, we define $\hat{\bm{d}}_{\mathrm{f},k} \coloneqq \bm{0}$ as the estimate for the future disturbances.
Further, we employ the disturbance data estimate $\hat{\mathcal{D}}_{T}$ defined as
\begin{equation}
    \hankel{1}{\hat{\mathcal{D}}_{T}} = \hankel{1}{\mathcal{Y}_{T}} \bm{\Pi}^{v}. \label{eq:distdata_estimate}
\end{equation}
Note that \eqref{eq:distdata_estimate} corresponds to the least-squares estimate presented by \cite{pan2021stochastic}, and also satisfies the consistency constraint~\eqref{eq:consistency}.
For ease of exposition, let us assume that the estimate $\hat{\mathcal{D}}_{T}$ satisfies the PE condition from Asm.~\ref{asm:trajData}(a).\footnote{If the estimate \eqref{eq:distdata_estimate} does not satisfy the PE condition from Asm.~\ref{asm:trajData}(a), alternative estimates can be found by exploring the solutions of the consistency constraint \eqref{eq:consistency}. Correspondingly, the estimation error relation \eqref{eq:sysparam_error} and subsequent prediction error analysis can be adapted to the new estimate. Note that, by Asm.~\ref{asm:trajData}, there exists at least one candidate estimate $\hat{\mathcal{D}}_{T}$ that satisfies the consistency constraint \eqref{eq:consistency} and Asm.~\ref{asm:trajData}(a), namely the unknown ${\mathcal{D}}_{T}$.}
By replacing the true disturbance data $\mathcal{D}_T$ and future disturbances $\bm{d}_{\mathrm{f},k}$ in the exact predictors \eqref{eq:predictor_output_exact} and \eqref{eq:predictor_extstate_exact} with the aforementioned estimates $\hat{\mathcal{D}}_{T}$ and $\hat{\bm{d}}_{\mathrm{f},k}$, we obtain the nominal output predictor
\begin{equation}
    \hat{\bm{y}}_{\mathrm{f},k} = \hankel{T_{\mathrm{f}}}{\mathcal{Y}_{T}}\bm{M}(\hat{\mathcal{D}}_{T})^{\dagger} \col{\bm{\xi}_k, \bm{v}_{\mathrm{f},k}, \bm{0}} \label{eq:predictor_output_estimate}
\end{equation}
and the corresponding extended-state predictor
\begin{equation}
    \hat{\bm{\xi}}_{\mathrm{f},k} = \bm{H}_{\mathrm{f},\xi}\bm{M}(\hat{\mathcal{D}}_{T})^{\dagger} \col{\bm{\xi}_k, \bm{v}_{\mathrm{f},k},\bm{0}}.\label{eq:predictor_extstate_estimate}
\end{equation}
Here, the vectors $\hat{\bm{y}}_{\mathrm{f},k} \coloneqq \col{\bm{y}_{0|k},\ldots,\bm{y}_{T_{\mathrm{f}}-1|k}}$ and $\hat{\bm{\xi}}_{\mathrm{f},k} \coloneqq \col{\bm{\xi}_{1|k},\ldots,\bm{\xi}_{T_{\mathrm{f}}|k}}$ entail the predicted outputs and extended-states, respectively, where the subscript ${i|k}$ denotes predictions $i$ steps ahead of the current step $k$.

Via Prop.~\ref{prop:consistency}~and~\ref{prop:consistency_multistep}, the relations \eqref{eq:connection_sysparams_data} and \eqref{eq:connection_sysparams_data_multistep_2} allow us to obtain system parameter estimates corresponding to the disturbance data estimate \eqref{eq:distdata_estimate}, denoted by the matrices $\hat{\bm{\Phi}}_{\mathrm{cl}}$ and $\hat{\bm{\Psi}}$ for the single-step form \eqref{eq:system}, and $\hat{\mathcal{O}}_{T_{\mathrm{f}}}$, $\hat{\mathcal{T}}_{T_{\mathrm{f}}}^{u}$, and $\hat{\mathcal{T}}_{T_{\mathrm{f}}}^{d}$ for the (causal) multi-step form \eqref{eq:system_multistep} of the system.
Specifically, from \eqref{eq:connection_sysparams_data}, we obtain the parameter estimate
\begin{equation}
    \mat{\hat{\bm{\Phi}}_{\mathrm{cl}} & \hat{\bm{\Psi}}} = \hankel{1}{\mathcal{Y}_{T}} \col{\hankel{1}{\mathcal{X}_{T}},\hankel{1}{\mathcal{V}_{T}}}^{\dagger}. \label{eq:sysparam_estimate}
\end{equation}
Furthermore, due to Prop.\;\ref{prop:consistency_multistep}, the relation \eqref{eq:connection_sysparams_data_multistep_2} implies that the nominal output predictor \eqref{eq:predictor_output_estimate} is equivalent to
\begin{equation}
    \hat{\bm{y}}_{\mathrm{f},k} = \hat{\mathcal{O}}_{T_{\mathrm{f}}} \bm{\xi}_k + \hat{\mathcal{T}}_{T_{\mathrm{f}}}^{v}\bm{v}_{\mathrm{f},k}.\label{eq:connection_predictor_output_estimate_model}
\end{equation}
In other words, a consequence of Prop.\;\ref{prop:consistency_multistep} is that the lower-triangular structure of the Toeplitz matrices in \eqref{eq:connection_sysparams_data_multistep_2} enforces causality onto the nominal predictors \eqref{eq:predictor_output_estimate} and \eqref{eq:predictor_extstate_estimate}, i.e., the predictions $i$ steps ahead of $k$ only depend on the initial state $\bm{\xi}_k$ and inputs $\bm{v}_{0|k},\ldots,\bm{v}_{i|k}$. By contrast, when neglecting the consistency constraints \eqref{eq:consistency} and \eqref{eq:consistency_multistep} on the estimate $\hat{\mathcal{D}}_T$, relation \eqref{eq:connection_sysparams_data_multistep_2} (i.e., the causal structure of the predictor) is not guaranteed. As demonstrated before, this causal structure is naturally given in the exact predictors \eqref{eq:predictor_output_exact} and \eqref{eq:predictor_extstate_exact} via the relations \eqref{eq:connection_sysparams_data_multistep}, \eqref{eq:connection_sysparams_data_multistep_2}.

Since we are leveraging \textit{estimates} for the nominal prediction \eqref{eq:predictor_output_estimate}, \eqref{eq:predictor_extstate_estimate}, a deviation between the estimates $\hat{\mathcal{D}}_T$, $\hat{\bm{d}}_{\mathrm{f},k}$ and the true (unknown) values $\mathcal{D}_T$, $\bm{d}_{\mathrm{f},k}$ will eventually lead to an error between the exact predictions \eqref{eq:predictor_output_exact}, \eqref{eq:predictor_extstate_exact} and the nominal predictions. Thus, in order to safely incorporate the nominal predictors \eqref{eq:predictor_output_estimate}, \eqref{eq:predictor_extstate_estimate} in controller design, we need to quantify and leverage this prediction error. To our advantage, the choice \eqref{eq:distdata_estimate} of the disturbance data estimates $\hat{\mathcal{D}}_{T}$ allows us to obtain a simple relation between the data uncertainty (in terms of the unknown data $\mathcal{D}_{T}$ satisfying Asm.~\ref{asm:disturb} and \ref{asm:trajData}) and the error between the true system parameters in \eqref{eq:system} and the corresponding estimate~\eqref{eq:sysparam_estimate}: Since the true (unknown) disturbance data $\mathcal{D}_{T}$ satisfies the consistency constraint \eqref{eq:consistency}, we can apply the relation \eqref{eq:connection_sysparams_data} to express the true system parameters as
\begin{equation}
    \mat{{\bm{\Phi}}_{\mathrm{cl}} & {\bm{\Psi}}} = \left(\hankel{1}{\mathcal{Y}_{T}}-\hankel{1}{{\mathcal{D}}_{T}}\right) \mat{\hankel{1}{\mathcal{X}_{T}}\\\hankel{1}{\mathcal{V}_{T}}}^{\dagger}. \label{eq:connection_truesysparams_distdata}
\end{equation}
By substracting the estimate \eqref{eq:sysparam_estimate} from \eqref{eq:connection_truesysparams_distdata}, we obtain the corresponding estimation error
\begin{equation}
    \mat{\bm{\Phi}_{\mathrm{cl}}-\hat{\bm{\Phi}}_{\mathrm{cl}} &~ \bm{\Psi}-\hat{\bm{\Psi}}} = -\hankel{1}{{\mathcal{D}}_{T}} \mat{\hankel{1}{\mathcal{X}_{T}}\\\hankel{1}{\mathcal{V}_{T}}}^{\dagger}. \label{eq:sysparam_error}
\end{equation}
With \eqref{eq:sysparam_error}, we can map statistical properties of the unknown disturbance data $\mathcal{D}_{T}$ (see Asm.~\ref{asm:disturb}) to the estimation error in terms of the system parameters. This will be useful in the subsequent section, where we analyze statistical properties of the prediction error mentioned above.

\subsection{Prediction Error Analysis \& Constraint Tightening} \label{sec:prederr}
In this section, we analyze the error between the exact predictions~\eqref{eq:predictor_output_exact} and~\eqref{eq:predictor_extstate_exact} and the proposed nominal predictions~\eqref{eq:predictor_output_estimate} and~\eqref{eq:predictor_extstate_estimate} for given initial condition $\tilde{\bm{\xi}}_k$ and predicted inputs $\hat{\bm{v}}_{\mathrm{f},k} \coloneqq \col{\bm{v}_{0|k},\,\ldots,\,\bm{v}_{T_{\mathrm{f}}-1|k}}$. This will then be used to define tightened nominal constraints as a tractable reformulation of the chance constraints~\eqref{eq:constraints}.

Since the exact output prediction~\eqref{eq:predictor_output_exact} corresponds to~\eqref{eq:system_multistep}, and the nominal output prediction~\eqref{eq:predictor_output_estimate} corresponds to \eqref{eq:connection_predictor_output_estimate_model}, the prediction error $\bm{e}^y_{\mathrm{f},k} \coloneqq \bm{y}_{\mathrm{f},k} - \hat{\bm{y}}_{\mathrm{f},k}$ evolves as
\begin{align}
    \bm{e}^y_{\mathrm{f},k} &= (\mathcal{O}_{T_{\mathrm{f}}} - \hat{\mathcal{O}}_{T_{\mathrm{f}}})  \tilde{\bm{\xi}}_k + (\mathcal{T}_{T_{\mathrm{f}}}^{v} - \hat{\mathcal{T}}_{T_{\mathrm{f}}}^{v})  \hat{\bm{v}}_{\mathrm{f},k} + \mathcal{T}_{T_{\mathrm{f}}}^{d} \bm{d}_{\mathrm{f},k}. \label{eq:prederror_output}
\end{align}
This error can be further decomposed into $\bm{e}^y_{\mathrm{f},k} = \bm{e}^{y,\text{data}}_{\mathrm{f},k} + \bm{e}^{y,\text{add}}_{\mathrm{f},k}$, with $\bm{e}^{y,\text{data}}_{\mathrm{f},k} = \col{\bm{e}^{y,\text{data}}_{0|k},\,\ldots,\,\bm{e}^{y,\text{data}}_{T_{\mathrm{f}}-1|k}} \coloneqq (\mathcal{O}_{T_{\mathrm{f}}} - \hat{\mathcal{O}}_{T_{\mathrm{f}}})  \tilde{\bm{\xi}}_k + (\mathcal{T}_{T_{\mathrm{f}}}^{v} - \hat{\mathcal{T}}_{T_{\mathrm{f}}}^{v})  \hat{\bm{v}}_{\mathrm{f},k}$, entailing the prediction error due to the mismatch between the true (unknown) disturbance data $\mathcal{D}_T$ and the disturbance data estimate $\hat{\mathcal{D}}_T$, and $\bm{e}^{y,\text{add}}_{\mathrm{f},k} = \col{\bm{e}^{y,\text{add}}_{0|k},\dots,\bm{e}^{y,\text{add}}_{T_{\mathrm{f}}-1|k}}\coloneqq\mathcal{T}_{T_{\mathrm{f}}}^{d} \bm{d}_{\mathrm{f},k}$ entailing the influence of the additive disturbances over the prediction horizon.
Considering the difference between the exact extended-state prediction~\eqref{eq:predictor_extstate_exact} and the nominal prediction~\eqref{eq:predictor_extstate_estimate}, we can leverage $\bm{\Phi}_{\mathrm{cl}} = \bm{E}_y\bm{A}_{\mathrm{cl}}$, $\bm{y}_k = \bm{E}_y \bm{\xi}_{k+1}$, $\bm{u}_k = \bm{E}_u \bm{\xi}_{k+1}$, and \eqref{eq:prederror_output} to obtain error dynamics for the extended-state prediction errors $\bm{e}^\xi_{l+1|k} \coloneqq \bm{\xi}_{k+l+1} - \bm{\xi}_{l+1|k}$ and the input prediction errors $\bm{e}^{u}_{l|k} \coloneqq \bm{u}_{k+l} - \bm{u}_{l|k}$ with $l\in\mathbb{N}_0^{T_{\mathrm{f}}-1}$.
Accordingly, similar to \eqref{eq:prederror_output},
we also decompose the input prediction error into $\bm{e}^{u}_{l|k} = \bm{e}^{u,\text{data}}_{l|k} + \bm{e}^{u,\text{add}}_{l|k}$.

In what follows, we first derive statistical properties of the individual error terms $\bm{e}^{y,\text{data}}_{l|k}$ and $\bm{e}^{y,\text{add}}_{l|k}$ (respectively, $\bm{e}^{u,\text{data}}_{l|k}$ and $\bm{e}^{u,\text{add}}_{l|k}$), and then combine the derived results to obtain a deterministic reformulation of the chance constraints~\eqref{eq:constraints}. Since the true system parameters $\bm{\Phi}$ and $\bm{\Psi}$ in~\eqref{eq:system} are unknown, we leverage the system parameter bounds from Asm.~\ref{asm:strongstab} in the analysis.

\subsubsection{Prediction Error From Disturbance Data Uncertainty:}
The next result yields stochastic properties of the prediction errors $\bm{e}^{y,\text{data}}_{l|k}$ and $\bm{e}^{u,\text{data}}_{l|k}$, $l\in\mathbb{N}_0^{T_{\mathrm{f}}-1}$, stemming from the deviation between the unknown disturbance data $\mathcal{D}_T$ and the disturbance data estimate $\hat{\mathcal{D}}_T$, cf. \eqref{eq:sysparam_error}. In particular, due to the constraint formulation in \eqref{eq:constraints}, we are interested in a (probabilistic) bound on $[\bm{G}_y]_i \bm{e}^{y,\text{data}}_{l|k}$ and $ [\bm{G}_u]_j\bm{e}^{u,\text{data}}_{l|k}$ for all $i \in \mathbb{N}_1^{r_y}$, $j \in \mathbb{N}_1^{r_u}$, and $l \in \mathbb{N}_0^{T_{\mathrm{f}}-1}$.
\begin{lemma} \label{lem:prederror_data}
    Consider fixed $\hat{\bm{v}}_{\mathrm{f},k} = \col{\bm{v}_{0|k},\,\ldots,\,\bm{v}_{T_{\mathrm{f}}-1|k}}$ and $\tilde{\bm{\xi}}_k \in \mathbb{X}$, and let 
    \begin{equation}
        \norm{\mat{\hankel{1}{\mathcal{X}_{T}}\\\hankel{1}{\mathcal{V}_{T}}}^{\dagger}\vc{\bm{\xi}_{l|k}\\\bm{v}_{l|k}}} \le \gamma_v ~~~ \forall l \in\mathbb{N}_0^{T_{\mathrm{f}}-1},~~\bm{\xi}_{0|k} \coloneqq \tilde{\bm{\xi}}_k \label{eq:inputstatebound_predicted}
    \end{equation}
    hold, with the predicted states $\bm{\xi}_{l|k}$ from \eqref{eq:predictor_extstate_estimate} and $\mathbb{X}$ and $\gamma_v$ from Asm.~\ref{asm:strongstab}(c).
    Additionally, let $c_{i}^{y}$ and $c_{j}^{u}$ be chosen such that $\lVert [\bm{G}_y]_i \bm{E}_y \bm{A}_{\mathrm{cl}}^k \bm{E}_y^{\top} \bm{\Sigma}_d^{\nicefrac{1}{2}} \rVert \le c_{i}^{y} \gamma_\xi^k$ and $\lVert [\bm{G}_u]_j \bm{E}_u \bm{A}_{\mathrm{cl}}^k \bm{E}_y^{\top} \bm{\Sigma}_d^{\nicefrac{1}{2}} \rVert \le c_{j}^{u} \gamma_\xi^k$ for all $i \in \mathbb{N}_1^{r_y}$, $j \in \mathbb{N}_1^{r_u}$, and $k\ge 0$ (cf. Asm.~\ref{asm:strongstab}(b)), where $\bm{\Sigma}_d^{\nicefrac{1}{2}}$ is such that $\bm{\Sigma}_d^{\nicefrac{1}{2}} \bm{\Sigma}_d^{\nicefrac{1}{2}} = \bm{\Sigma}_d$. Then, for every $\tilde{\eps}^y, \tilde{\eps}^u \in (0,1)$ and $l \in\mathbb{N}_0^{T_{\mathrm{f}}-1}$ the prediction error terms $[\bm{G}_y]_i\bm{e}^{y,\text{data}}_{l|k}$ and $[\bm{G}_u]_j\bm{e}^{u,\text{data}}_{l|k}$ satisfy
    \begin{subequations}\label{eq:dataerror_var}
    	\begin{align}
    	\Pr{\left[\bm{G}_y\right]_i \bm{e}^{y,\text{data}}_{l|k}  \le c_{i}^{y} \beta_l(\tilde{\eps}^y)~\big|~\tilde{\bm{\xi}}_{k}}\, &\ge 1-\tilde{\eps}^y,\label{eq:outputerror_var} \\
    	\Pr{\left[\bm{G}_u\right]_j \bm{e}^{u,\text{data}}_{l|k} \le c_{j}^{u}\beta_l(\tilde{\eps}^u)~\big|~\tilde{\bm{\xi}}_{k}}\, &\ge 1-\tilde{\eps}^u,\label{eq:inputerror_var}
    	\end{align}
    \end{subequations}
    with $\beta_l(\eps) \coloneqq \big(\textstyle\sum_{\iota=0}^{l} \gamma_{\xi}^{\iota}\big) \gamma_v \sqrt{n_{\theta}(1+f^{-1}(\eps^{-\frac{2}{n_{\theta}}}))}$, $n_{\theta}\coloneqq n_y(n_{\xi}+n_u)$, and $f(\delta) \coloneqq \exp{\delta}/(1+\delta)$.
      \hfill \small{$\square$}
\end{lemma}
\begin{proof}
    See \ref{app:proof_prederror_data}.
    \hfill \small{$\blacksquare$}
\end{proof}

We remark that condition \eqref{eq:inputstatebound_predicted} is imposed on the deterministic predicted states $\bm{\xi}_{l|k}$, $l\in\mathbb{N}_0^{T_{\mathrm{f}}-1}$, computed via the identified predictor \eqref{eq:predictor_extstate_estimate}, and will later be guaranteed by the design of the predictive control scheme. Additionally, valid parameters $ c_{i}^{y}$, $ c_{j}^{u}$ are readily obtained from Asm.~\ref{asm:disturb}~and~\ref{asm:strongstab}(b) as $c_{i}^{y} = \norm{[\bm{G}_y]_i}c_{\xi}{\ol{\sigma}}_d$ and $c_{j}^{u} = \norm{[\bm{G}_u]_j}c_{\xi}{\ol{\sigma}}_d$, leveraging the fact that $\norm{\bm{E}_y} = \norm{\bm{E}_u} = 1$. However, this choice might be overly conservative. In \ref{app:strongstab}, we discuss how less conservative $ c_{i}^{y}$, $ c_{j}^{u}$ can be obtained. Lastly, note that the bounds in \eqref{eq:dataerror_var} scale with $\sqrt{n_{\theta}}$.

\subsubsection{Prediction Error From Additive Disturbances:}
The following result yields stochastic properties of the prediction errors $\bm{e}^{y,\text{add}}_{l|k}$ and $\bm{e}^{u,\text{add}}_{l|k}$, $l\in\mathbb{N}_0^{T_{\mathrm{f}}-1}$, stemming from the influence of the future additive disturbances $\bm{d}_{\mathrm{f},k}$.
\begin{lemma} \label{lem:prederror_add}
    Let $\mathbb{A}_{\mathrm{cl}}^{\mathrm{d}}$ be the set of matrices $\bm{A}_{\mathrm{cl}} = \bm{A} + \bm{B}\bm{K}$ such that $(\bm{A},\bm{B}) \in \mathbb{A}^{\mathrm{d}}$ from Asm.\;\ref{asm:strongstab}(a).
    The prediction error terms $[\bm{G}_y]_i\bm{e}^{y,\text{add}}_{l|k}$ and $[\bm{G}_u]_j\bm{e}^{u,\text{add}}_{l|k}$ follow a zero-mean sub-Gaussian distribution with variance proxies $\ol{\bm{\Sigma}}^{y,\text{add}}_{l,i} \coloneqq [\bm{G}_y]_i\bm{E}_y \ol{\bm{\Sigma}}^{\xi,\text{add}}_{l+1} \bm{E}_y^{\top} [\bm{G}_y]_i^{\top} $ and $\ol{\bm{\Sigma}}^{u,\text{add}}_{l,j} \coloneqq [\bm{G}_u]_j\bm{E}_u\ol{\bm{\Sigma}}^{\xi,\text{add}}_{l+1} \bm{E}_u^{\top} [\bm{G}_u]_j^{\top}$, respectively, for all $i \in \mathbb{N}_1^{r_y}$, $j \in \mathbb{N}_1^{r_u}$, and $l \in\mathbb{N}_0^{T_{\mathrm{f}}-1}$, with $\ol{\bm{\Sigma}}^{\xi,\text{add}}_{0} \coloneqq \bm{0}$ and
	\begin{align}
		\left(\ol{\bm{\Sigma}}^{\xi,\text{add}}_{l+1}\right)^{-1}& \coloneqq \underset{\bm{X}\succ \bm{0}}{\arg\min} -\log\det \bm{X}^{-1} \label{eq:extstateerror_add_var}\\		
		\mathrm{s.t.}~~ &  \bm{X} - \bm{A}_{\mathrm{cl}} \ol{\bm{\Sigma}}^{\xi,\text{add}}_{l} \bm{A}_{\mathrm{cl}}^{\top} - \bm{\Sigma}_d^{\text{ext}}  \succ \bm{0}~~~ \forall \bm{A}_{\mathrm{cl}}\in\mathbb{A}^{\mathrm{d}}_{\mathrm{cl}}, \notag 
    \end{align}
    where $\bm{\Sigma}_d^{\text{ext}}$ is chosen such that $\bm{\Sigma}_d^{\text{ext}} -\bm{E}_y^{\top} \bm{\Sigma}_d \bm{E}_y \succ \bm{0}$.
    \hfill \small{$\square$}
\end{lemma}
\begin{proof}
    See \ref{app:proof_prederror_add}.
    \hfill \small{$\blacksquare$}
\end{proof}
In contrast to Lem.~\ref{lem:prederror_data}, Lem.~\ref{lem:prederror_add} leverages sub-Gaussianity of this prediction error term and employs over-approximations of the variance proxies in terms of the Loewner
order, thus minimizing the spread in the direction of the
principal components \citep{boyd2004convex,arcari2023stochastic}.
Note that \eqref{eq:extstateerror_add_var} can be reformulated as a convex semi-definite program via \cite[Lemma~3]{arcari2023stochastic} and is thus efficiently solvable.

\subsubsection{Reformulation of Chance Constraints:} We now exploit the derived stochastic properties of the prediction errors $\bm{e}^y_{l|k}$ and $\bm{e}^u_{l|k}$ (see Lem.~\ref{lem:prederror_data} and~\ref{lem:prederror_add}) to obtain deterministic reformulations of the chance constraints~\eqref{eq:constraints}.
The chance constraints are reformulated in terms of tightened nominal constraints for the predictions $\bm{y}_{l|k}$ and $\bm{u}_{l|k}$, i.e., 
\begin{subequations}\label{eq:tightcons}
 \begin{align} 
    \bm{G}_y \bm{y}_{l|k} &\le \bm{g}_y - \bm{\eta}_l^y~~~~\forall l \in \mathbb{N}_0^{T_{\mathrm{f}}-1},\\
    \bm{G}_u \bm{u}_{l|k} &\le \bm{g}_u - \bm{\eta}_l^u~~~~\forall l \in \mathbb{N}_0^{T_{\mathrm{f}}-1},
\end{align}   
\end{subequations}
with parameters $\bm{\eta}_l^y$, $\bm{\eta}_l^u$, specified via the following result.
\begin{lemma} \label{lem:constight}
    Consider fixed $\hat{\bm{v}}_{\mathrm{f},k} = \col{\bm{v}_{0|k},\,\ldots,\,\bm{v}_{T_{\mathrm{f}}-1|k}}$ and $\tilde{\bm{\xi}}_k \in \mathbb{X}$, $\ol{\sigma}_d = \norm{\Sigma_d}$ from Asm.\;\ref{asm:disturb}, and let condition~\eqref{eq:inputstatebound_predicted} hold. Furthermore, consider the predicted outputs $\hat{\bm{y}}_{\mathrm{f},k} = \col{\bm{y}_{0|k},\,\dots,\,\bm{y}_{T_{\mathrm{f}}-1|k}}$ from \eqref{eq:predictor_output_estimate}, the predicted inputs $\hat{\bm{u}}_{\mathrm{f},k} = \col{\bm{u}_{0|k},\,\dots,\,\bm{u}_{T_{\mathrm{f}}-1|k}}$ via \eqref{eq:inputdecomp} and \eqref{eq:predictor_extstate_estimate}, and the tightening parameters $\bm{\eta}_l^y$ and $\bm{\eta}_l^u$ defined as
\begin{subequations} \label{eq:tightening_params}
    \begin{align}
        [\bm{\eta}_l^y]_i &\coloneqq c_{i}^{y}\beta_l(\eps_j^y/2)+ \sqrt{2\ln(2/\eps^y_{i})\ol{\bm{\Sigma}}^{y,\text{add}}_{l,i}}~~\forall i \in \mathbb{N}_1^{r_y}, \label{eq:outputcons_tightening}\\
        [\bm{\eta}_l^u]_j &\coloneqq c_{j}^{u}\beta_l(\eps_i^u/2) + \sqrt{2\ln(2/\eps^u_{j}) \ol{\bm{\Sigma}}^{u,\text{add}}_{l,j}}\;\forall j \in \mathbb{N}_1^{r_u}.\label{eq:inputcons_tightening}
    \end{align}
\end{subequations}
    If the predicted input $\bm{u}_{l|k}$ and output $\bm{y}_{l|k}$ satisfy the tightened nominal constraints \eqref{eq:tightcons}, then
    \begin{subequations} \label{eq:chancecons_pred}
        \begin{align}
            \Pr{\left[\bm{G}_y\right]_i \bm{y}_{k+l} \le \left[\bm{g}_y\right]_i ~\big|~ \tilde{\bm{\xi}}_k}\, &\ge 1-\eps^y_{i} ~~\forall i\in\mathbb{N}_1^{r_y},\\
            \Pr{\left[\bm{G}_u\right]_j \bm{u}_{k+l} \le \left[\bm{g}_u\right]_j~\big|~\tilde{\bm{\xi}}_k} &\ge 1-\eps^u_{j}~~\forall j\in\mathbb{N}_1^{r_u},
        \end{align}
    \end{subequations}
    i.e., the true input $\bm{u}_{k+l}$ and output $\bm{y}_{k+l}$ satisfy the chance constraints~\eqref{eq:constraints} conditioned on $\tilde{\bm{\xi}}_k$, with $ l \in \mathbb{N}_0^{T_{\mathrm{f}}-1}$.
    \hfill \small{$\square$}
\end{lemma}
\begin{proof}
    The claim follows from \cite[Thm.~1]{ao2025stochastic} and \cite[Lem.~2]{ao2025stochastic}; see \ref{app:proof_constight}.
    \hfill \small{$\blacksquare$}
\end{proof}
Although Lem.~\ref{lem:constight} is formulated based on individual chance constraints (see \eqref{eq:constraints}), we remark that joint chance constraints can readily be considered using Boole's inequality, cf. \citep[Cor.~7]{knaup2024recursively}. Also, note that \eqref{eq:chancecons_pred} is consistent to the notation in \eqref{eq:constraints} when defining $\tilde{k}$ as the time step for which $\bm{\xi}_{\tilde{k}} = \tilde{\bm{\xi}}_k$ (e.g., $\tilde{k}=0$, $\tilde{\bm{\xi}}_k = \bm{\xi}_0$).

We are now equipped with data-driven predictors for the system behavior via \eqref{eq:predictor_output_estimate} and \eqref{eq:predictor_extstate_estimate}, and with corresponding tightened constraints~\eqref{eq:tightcons} as a tractable approximation of the chance constraints~\eqref{eq:constraints} via Lem.~\ref{lem:constight}. Thus, combining the results from the previous sections allows us to approximate OCP \eqref{eq:ocp_original} and present the proposed DPC scheme.

\subsection{Proposed Stochastic DPC Scheme} \label{sec:sdpc}
We now elaborate on the design steps of the proposed stochastic DPC scheme and its theoretical guarantees. 

Since we aim for a finite-horizon approximation of the OCP \eqref{eq:ocp_original}, typically a terminal cost $\norm{\bm{\xi}_{T_{\mathrm{f}}|k}}^2_{\bm{P}}$ with weight matrix $\bm{P} = \bm{P}^{\top} \succ \bm{0}$ is employed to approximate the infinite-horizon tail. Given that a stability analysis is outside the scope of this work, we do not specify the design of $\bm{P}$. In related settings, the weight $\bm{P}$ is commonly chosen such that $\bm{P} - \bm{A}^{\top}_{\mathrm{cl}} \bm{P} \bm{A}_{\mathrm{cl}} - \bm{K}^{\top} \bm{R} \bm{K} - \bm{A}^{\top}_{\mathrm{cl}} \bm{E}_y^{\top} \bm{Q} \bm{E}_y^{\top} \bm{A}_{\mathrm{cl}} \succ \bm{0}$ is satisfied for all $\bm{A}_{\mathrm{cl}} = \bm{A} + \bm{B}\bm{K}$, $(\bm{A},\bm{B}) \in \mathbb{A}^{\mathrm{d}}$, see Asm.~\ref{asm:strongstab}(a) \citep{arcari2023stochastic,teutsch2024sampling}.

However, we are interested in closed-loop guarantees on recursive feasibility and constraint satisfaction, for which we rely on the following assumption on terminal constraints.

\begin{assumption} \label{asm:termIngredients}
    There exists a terminal constraint set $\mathbb{X}_{\infty}$ that satisfies the following for all $\bm{\xi} \in \mathbb{X}_{\infty}$:
    \begin{itemize}
        \item[(a)] $\hat{\bm{A}}_{\mathrm{cl}}\bm{\xi} \in \mathbb{X}_{\infty}$, with $\bm{A}_{\mathrm{cl}}$ corresponding to the disturbance data estimate $\hat{\mathcal{D}}_{T}$ \eqref{eq:distdata_estimate} via  \eqref{eq:sysparam_estimate} and \eqref{eq:system_arx_nonminimal_helper},
        \item[(b)] $\bm{G}_y \bm{E}_y\bm{\xi} \le \bm{g}_y - \bm{\eta}_{\infty}^y$ and $\bm{G}_u \bm{E}_u\bm{\xi} \le \bm{g}_u - \bm{\eta}_{\infty}^u$, with $\bm{\eta}_{\infty}^y \ge \bm{\eta}_{l}^y$ and $\bm{\eta}_{\infty}^u \ge \bm{\eta}_{l}^u$ from Lem.\;\ref{lem:constight} for all $l \ge 0$. \hfill \small{$\square$}
    \end{itemize}
\end{assumption}
\begin{remark}
    Variants of Asm.\;\ref{asm:termIngredients} are common in robust and stochastic predictive control literature to guarantee recursive feasibility \citep{arcari2023stochastic,ao2025stochastic,teutsch2024sampling}: (a) entails invariance under the feedback $\bm{u} = \bm{K}\bm{\xi}$ and the model estimate corresponding to the disturbance data estimate $\hat{\mathcal{D}}_{T}$ employed in the predictors~\eqref{eq:predictor_output_estimate} and \eqref{eq:predictor_extstate_estimate}, whereas (b) entails satisfaction of the tightened input and output constraints from Lem.\;\ref{lem:constight}. The parameters $\bm{\eta}_{\infty}^y$ and $\bm{\eta}_{\infty}^u$ in (b) can be computed via Lem.\;\ref{lem:constight}, leveraging \eqref{eq:dataerror_var} and \eqref{eq:extstateerror_add_var} for $l \to \infty$. Note that for $l \to \infty$ in \eqref{eq:dataerror_var}, we have the geometric sum $\textstyle\sum_{i=0}^{\infty} \gamma_{\xi}^i = (1-\gamma_{\xi})^{-1}$. A corresponding set $\mathbb{X}_{\infty}$ can be determined using standard methods \cite[Sec.~5.3]{blanchini2015set}. \hfill \small{$\square$}
\end{remark}

Using Asm.\;\ref{asm:termIngredients}, we approximate the conceputal OCP~\eqref{eq:ocp_original} by developing a DPC scheme based on the nominal predictors~\eqref{eq:predictor_output_estimate} and \eqref{eq:predictor_extstate_estimate} and tightened constraints~\eqref{eq:tightcons}. Using the predictor matrices $\hat{\bm{M}}_y \coloneqq \hankel{T_{\mathrm{f}}}{\mathcal{Y}_{T}} \bm{M}(\hat{\mathcal{D}}_T)^{\dagger}$ and $\hat{\bm{M}}_{\xi} \coloneqq \bm{H}_{\mathrm{f},\xi} \bm{M}(\hat{\mathcal{D}}_T)^{\dagger}$ obtained from the identified output and extended-state predictors \eqref{eq:predictor_output_estimate} and \eqref{eq:predictor_extstate_estimate}, the OCP associated with the proposed DPC scheme is
\begin{subequations} \label{eq:ocp}
		\begin{align}
		& \underset{\bm{v}_{\mathrm{f},k}}{\mathrm{minimize}} \sum\limits_{l=0}^{T_{\mathrm{f}}-1} \left( \norm{\check{\bm{y}}_{l|k}}^2_{\bm{Q}} + \norm{\check{\bm{u}}_{l|k}}^2_{\bm{R}} \right) + \norm{\check{\bm{\xi}}_{T_{\mathrm{f}}|k}}^2_{\bm{P}}&\label{eq:ocp_cost}\\	
		&\mathrm{s.t.}~~  ~~~~\check{\bm{\xi}}_{0|k} = \bm{\xi}_k, ~~~~~ {\bm{\xi}}_{0|k} = \tilde{\bm{\xi}}_k ~\in\mathbb{X}\label{eq:ocp_init}\\
        & \phantom{\mathrm{s.t.}~~} \col{\check{\bm{y}}_{0|k}, \ldots,\check{\bm{y}}_{T_{\mathrm{f}}-1|k}} = \hat{\bm{M}}_y \col{\check{\bm{\xi}}_{0|k}, \bm{v}_{\mathrm{f},k}, \bm{0}} \label{eq:ocp_outputpred_cost} \\
        & \phantom{\mathrm{s.t.}~~ } ~~~\,\col{\check{\bm{\xi}}_{1|k}, \ldots,\check{\bm{\xi}}_{T_{\mathrm{f}}|k}} = \hat{\bm{M}}_{\xi} \col{\check{\bm{\xi}}_{0|k}, \bm{v}_{\mathrm{f},k}, \bm{0}} \label{eq:ocp_statepred_cost}\\
        & \phantom{\mathrm{s.t.}~~ } ~~~~\check{\bm{u}}_{l|k} = \bm{K}\check{\bm{\xi}}_{l|k} + \bm{v}_{l|k}  ~~~\, \forall l\in\mathbb{N}_0^{T_{\mathrm{f}}-1} \label{eq:ocp_inputpred_cost}\\        
		&\phantom{\mathrm{s.t.}~~}   \col{{\bm{y}}_{0|k}, \ldots,{\bm{y}}_{T_{\mathrm{f}}-1|k}} = \hat{\bm{M}}_y \col{{\bm{\xi}}_{0|k}, \bm{v}_{\mathrm{f},k}, \bm{0}} \label{eq:ocp_outputpred_constraint}\\
        & \phantom{\mathrm{s.t.}~~ } ~~~\,\col{{\bm{\xi}}_{1|k}, \ldots,{\bm{\xi}}_{T_{\mathrm{f}}|k}} = \hat{\bm{M}}_{\xi} \col{{\bm{\xi}}_{0|k}, \bm{v}_{\mathrm{f},k}, \bm{0}} \label{eq:ocp_statepred_constraint} \\
        & \phantom{\mathrm{s.t.}~~ } ~~~~ {\bm{u}}_{l|k} = \bm{K}{\bm{\xi}}_{l|k} + \bm{v}_{l|k}  ~~~\, \forall l\in\mathbb{N}_0^{T_{\mathrm{f}}-1} \label{eq:ocp_inputpred_constraint}\\
		& \phantom{\mathrm{s.t.}~~ } \bm{G}_y {\bm{y}}_{l|k} \le \bm{g}_y - \tilde{\bm{\eta}}^y_{l+k-\tilde{k}} ~~~\, \forall l\in\mathbb{N}_0^{T_{\mathrm{f}}-1} \label{eq:ocp_outputcons}\\
	    & \phantom{\mathrm{s.t.}~~ } \bm{G}_u {\bm{u}}_{l|k} \le \bm{g}_u - \tilde{\bm{\eta}}^u_{l+k-\tilde{k}} ~~~ \forall l\in\mathbb{N}_0^{T_{\mathrm{f}}-1} \label{eq:ocp_inputcons}\\
		&\phantom{\mathrm{s.t.}~~ } ~~\;\,{\bm{\xi}}_{T_{\mathrm{f}}|k} \in \mathbb{X}_{\infty}.& \label{eq:ocp_termcons}
		\end{align}
\end{subequations}
As in indirect feedback SMPC \citep{hewing2020recursively}, predictions for the cost function \eqref{eq:ocp_cost} are based on the measured state $\bm{\xi}_k$ (see \eqref{eq:ocp_outputpred_cost}--\eqref{eq:ocp_inputpred_cost}), whereas for the constraints \eqref{eq:ocp_outputcons}--\eqref{eq:ocp_termcons}, we use predictors \eqref{eq:ocp_outputpred_constraint}--\eqref{eq:ocp_inputpred_constraint} based on an artificial initialization $\tilde{\bm{\xi}}_k~\in\mathbb{X}$ in \eqref{eq:ocp_init}. Specifically, we choose $\tilde{\bm{\xi}}_k = \bm{\xi}_k$ if OCP~\eqref{eq:ocp} is feasible with this choice; else, we choose $\tilde{\bm{\xi}}_k = {\bm{\xi}}_{1|k-1}$ (i.e., the previously predicted nominal state). Regarding the chance constraint~\eqref{eq:constraints}, this allows us to define $\tilde{k}$ as the most recent time step $\iota \in \mathbb{N}_0^k$ for which OCP~\eqref{eq:ocp} was feasible using the measured state $\tilde{\bm{\xi}}_\iota = \bm{\xi}_\iota$. 
Accordingly, the constraint parameters $\tilde{\bm{\eta}}^y_{i}$, $\tilde{\bm{\eta}}^u_{i}$ are defined as $\tilde{\bm{\eta}}^y_{i} = \bm{\eta}^y_{i}$ and $\tilde{\bm{\eta}}^u_{i} = \bm{\eta}^u_{i}$ from Lem.\;\ref{lem:constight} for $i \in \mathbb{N}_0^{T_{\mathrm{f}}-1}$, and $\tilde{\bm{\eta}}^y_{i} = \bm{\eta}^y_{\infty}$ and $\tilde{\bm{\eta}}^u_{i} = \bm{\eta}^u_{\infty}$ from Asm.\;\ref{asm:termIngredients} for $i \ge T_{\mathrm{f}}$. This initialization strategy combines indirect and direct feedback formulations, similar to \cite[Cor.~8]{knaup2024recursively} or works on initial state optimization \citep{schluter2022stochastic,pan2023data}. Note that \eqref{eq:ocp_init} requires $\tilde{\bm{\xi}}_k {\in\mathbb{X}}$ for feasibility of the OCP. This renders condition~\eqref{eq:inputstatebound_predicted} valid via Asm.~\ref{asm:strongstab}(c) and \eqref{eq:ocp_outputcons}--\eqref{eq:ocp_inputcons}, thus enabling the probabilistic guarantees from Lem.~\ref{lem:constight}.

The implicit control law associated with OCP~\eqref{eq:ocp} reads
\begin{equation} \label{eq:controllaw}
    \bm{\kappa}\left(\bm{\xi}_k\right):= \bm{u}^{*}_k = \bm{K} \bm{\xi}_k + \bm{v}^*_{0|k},
\end{equation}
utilizing the first element of the optimal solution~$\bm{v}^*_{\mathrm{f},k}$. The overall control scheme is summarized in Algorithm\;\ref{alg:controller}. 
\begin{algorithm}
\caption{Stochastic DPC} \label{alg:controller}
\begin{algorithmic}[1]
\renewcommand{\algorithmicrequire}{\textbf{Offline Phase:}}
\renewcommand{\algorithmicensure}{\textbf{Online Phase:}}
\REQUIRE
\STATE Obtain data $\mathcal{U}_T$, $\mathcal{Y}_T$, $\mathcal{X}_{T+1}$ satisfying Asm.\;\ref{asm:trajData} and the disturbance data estimate $\hat{\mathcal{D}}_{T}$ via \eqref{eq:distdata_estimate}.
\STATE Determine $\bm{K}$, $c_{\xi}$, $\gamma_{\xi}$, and $\gamma_v$ satisfying Asm.\;\ref{asm:strongstab}.
\STATE Determine parameters $\bm{\eta}^y_{l}$, $\bm{\eta}^u_{l}$, $l\in \mathbb{N}_0^{T_{\mathrm{f}}-1}$ via Lem.\;\ref{lem:constight}.
\STATE Determine $\bm{\eta}^y_{\infty}$, $\bm{\eta}^u_{\infty}$, and $\mathbb{X}_{\infty}$ satisfying Asm.\;\ref{asm:termIngredients}.
\ENSURE (for all $k \ge 0$)
\STATE Construct $\bm{\xi}_k$ from measurements $\bm{u}_{i},\,\bm{y}_{i}$, $i \in \mathbb{N}_{k-T_{\mathrm{p}}}^{k-1}$, and determine the initialization $\tilde{\bm{\xi}}_k$.
\STATE Solve the OCP \eqref{eq:ocp} to obtain $\bm{v}^*_{\mathrm{f},k}$.
\STATE Apply the input $\bm{u}_{k} = \bm{K} \bm{\xi}_k + \bm{v}^*_{0|k}$ to the system~\eqref{eq:system}.
\end{algorithmic}
\end{algorithm}

Lastly, the following result formalizes the theoretical closed-loop guarantees of recursive feasibility and chance-constraint satisfaction conditioned on the measurement $\bm{\xi}_{\tilde{k}}$ i.e., the most recent feasible measured state, with $\tilde{k} \in  \mathbb{N}_0^k$.
\begin{theorem} \label{thm:properties}
    Let Asm.\;\ref{asm:minss}--\ref{asm:termIngredients} hold. If the OCP~\eqref{eq:ocp} is feasible for $k = 0$, then it is recursively feasible for all $k\ge 1$. Furthermore, for all $k \ge 0$, the closed-loop system~\eqref{eq:system} under control law~\eqref{eq:controllaw} satisfies the chance constraints~\eqref{eq:constraints} conditioned on the measurement $\bm{\xi}_{\tilde{k}}$, where $\tilde{k} \in \mathbb{N}_0^k$ is the most recent time step at which OCP~\eqref{eq:ocp} was feasible with the initialization $\bm{\xi}_{0|\tilde{k}} = \bm{\xi}_{\tilde{k}}$. \hfill \small{$\square$}
\end{theorem}
\begin{proof}
    The proof follows arguments presented in \cite[Thm.~1--2]{hewing2020recursively}, \cite[Cor.~8]{knaup2024recursively}, and \cite{ao2025stochastic}; see \ref{app:proof_properties}.
    \hfill \small{$\blacksquare$}
\end{proof}
We remark that, due to the employed initialization strategy inspired by \cite{knaup2024recursively}, the constraint satisfaction guarantee in Thm.~\ref{thm:properties} differs from classical results of direct and indirect feedback stochastic MPC \citep{hewing2020recursively}: Typically, chance constraints are conditioned on $\bm{\xi}_0$ for indirect feedback approaches, and on $\bm{\xi}_k$ for direct feedback approaches. In contrast, Thm.~\ref{thm:properties} guarantees satisfaction of chance constraints conditioned on $\bm{\xi}_{\tilde{k}}$, thus employing the direct feedback result whenever feasible. The indirect feedback result is readily recovered when restricting the initialization in \eqref{eq:ocp_init} to $\bm{\xi}_{0|k} = \bm{\xi}_{1|k-1}$.

\section{Evaluation} \label{sec:eval}
In this section, we numerically evaluate the proposed DPC scheme on a system of the form \eqref{eq:system} with parameters
\begin{equation}
\bm{\Phi} = \mat{0.939 & 1 & 0.073 \\ 
                 0.017 & -0.060 & 0.997},~~\bm{\Psi} = \bm{0}
\end{equation}
and $T_{\mathrm{p}} = 1$, adapted from \cite{teutsch2024sampling} by normalizing the input. The system is subject to the box constraints $\norminf{\bm{u}_k}\le 1$, $\norminf{\bm{y}_k}\le 3$, and subject to normally distributed disturbances with $\bm{\Sigma}_d = \mathrm{diag}(0.05^2,\,0.025^2)$.
We collect an input--output data trajectory of length $T=100$  satisfying Asm.\;\ref{asm:trajData} by applying random admissible inputs, and obtain the disturbance data estimate~\eqref{eq:distdata_estimate}. In order to satisfy Asm.\;\ref{asm:strongstab}(a), the disturbance distribution is truncated at two standard deviations for the data collection; the set $\mathbb{A}^{\mathrm{d}}$ is then obtained via the relation \eqref{eq:connection_sysparams_data} using the disturbance data bounds (see Remark~\ref{rem:strongstab}). For the cost function~\eqref{eq:ocp_cost}, we choose a prediction horizon of $T_{\mathrm{f}} = 10$ and weights $\bm{R} = 1$, $\bm{Q} = \mathrm{diag}(1,\,100)$.
The feedback gain $\bm{K}$ and constants $c_{i}^{y}$, $c_{j}^{u}$ for Lem.~\ref{lem:prederror_data} (and, subsequently, for the tightened constraints~\eqref{eq:tightcons}) are determined using Prop.~\ref{prop:strongstab} from \ref{app:strongstab} with $\gamma_{\xi} = 0.9$. Accordingly, the terminal constraint set $\mathbb{X}_{\infty}$ is chosen such that Asm.\;\ref{asm:termIngredients} is satisfied. The risk parameters are set to $\eps^y_{i} = \eps^u_{j} = 0.1$.

In order to properly evaluate the proposed controller, we implement four different variants of the proposed control algorithm (Algorithm~\ref{alg:controller}) based on OCP \eqref{eq:ocp}:
\begin{itemize}
	\item[(a)] \textit{Proposed}: The proposed controller as detailed in Algorithm~\ref{alg:controller},
	\item[(b)] \textit{Estimate}: Algorithm~\ref{alg:controller}, but the constraint tightening parameters $\bm{\eta}_l^u$ and $\bm{\eta}_l^y$ are computed by neglecting the disturbance data uncertainty (i.e., $c^{y}_i = c^{u}_j = 0$ and $\eps^{y}_i,\eps^{u}_j \to 2\eps^{y}_i,2\eps^{u}_j$ in \eqref{eq:tightening_params}), and by only using the estimate $\hat{\bm{A}}_{\mathrm{cl}}$ from Asm.~\ref{asm:termIngredients}(a) for \eqref{eq:extstateerror_add_var}, 
	\item[(c)] \textit{Exact}: Algorithm~\ref{alg:controller}, but based on the true disturbance data $\mathcal{D}_T$ (i.e., exact model knowledge), thus $c^{y}_i = c^{u}_j = 0$ and $\eps^{y}_i,\eps^{u}_j \to 2\eps^{y}_i,2\eps^{u}_j$ in \eqref{eq:tightening_params} and the true $\bm{A}_{\mathrm{cl}}$ is used for \eqref{eq:extstateerror_add_var} to compute $\bm{\eta}_l^u$ and $\bm{\eta}_l^y$,
    \item[(d)] \textit{OnlyIF}: Algorithm~\ref{alg:controller}, but strictly using the indirect feedback initialization $\tilde{\bm{\xi}}_k = \tilde{\bm{\xi}}_{1|k-1}$ in \eqref{eq:ocp_init} ($\tilde{k}=0$).
\end{itemize}

Starting at $\bm{\xi}_0 \coloneqq \col{0,\,0,\,2.5}$, the control goal is to steer the system to the origin while satisfying constraints~\eqref{eq:constraints}. In a Monte-Carlo simulation of $1\,000$ runs, each controller is applied for $30$ time steps. At each time step, a random disturbance affects the system. The simulations are carried out in MATLAB using the \texttt{quadprog} solver.

Fig.\;\ref{fig:traj} shows trajectories from 50 exemplary runs of the simulation for the proposed controller. The probabilistic constraint tightening of the proposed scheme enables the system to operate near the constraint boundary, resulting in rapid convergence.
For performance assessment, we compare the total trajectory cost $J_{\mathrm{tot}}= \textstyle \sum_{k=0}^{29} \big(\bm{y}_{k}^{\top} \bm{Q}  \bm{y}_{k} + \bm{u}_{k}^{\top}\bm{R}\bm{u}_{k} \big)$ relative to the exact case in Tab.~\ref{tab:cost}, and the empirical probability $\hat{\eps}_{1,k}$ of constraint violation regarding the constraint on the output $[\bm{y}_k]_1$ in Fig.\;\ref{fig:risk}. The empirical risk $\hat{\eps}_{1,k}$ is computed by counting the number of constraint violations per time step $k$ in each run of the simulation and dividing the result by $1\,000$ (i.e., the total number of runs). Although the certainty equivalence approach (\textit{Estimate}) yields the lowest mean cost, this comes at the price of more frequent constraint violations, even exceeding the specified threshold $\eps^y_{1} = 0.1$ at some time steps. In contrast, the proposed method results in no constraint violations due to the conservatism introduced by Lem.~\ref{lem:prederror_data} and~\ref{lem:constight}, thus yielding larger total trajectory costs. The method that relies purely on the indirect feedback initialization yields the highest trajectory costs. The reason for this is that for $k \ge T_{\mathrm{f}} = 10$, the constraint tightening is dominated by the conservative constraint tightening from the terminal set $\mathbb{X}_{\infty}$. In contrast, the proposed method avoids this by resetting the shift in the tightened constraints \eqref{eq:ocp_outputcons} and \eqref{eq:ocp_inputcons}  whenever the measured extended-state $\bm{\xi}_k$ is feasible.

\begin{figure}
    \centering
    \includegraphics[page=1, clip, trim=15mm 200mm 105mm 75mm, width=0.5\textwidth]{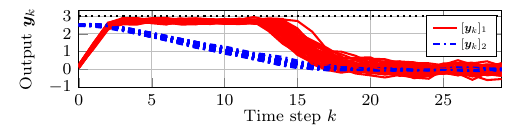}
    \caption{Trajectories of 50 runs under the proposed DPC scheme. Constraints are shown in dotted black lines.}
    \label{fig:traj}
\end{figure}
\begin{figure}
    \centering
    \includegraphics[page=1, clip, trim=15mm 77mm 105mm 200mm, width=0.5\textwidth]{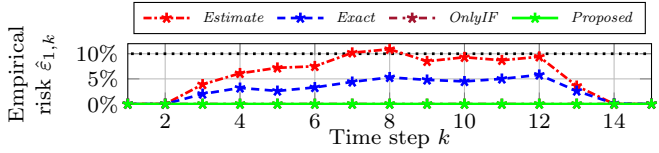}
    \caption{Empirical risk $\hat{\eps}_{1,k}^y$ of constraint violation for the output $[\bm{y}_k]_1$ over $1\,000$ runs. The dotted black line depicts the specified threshold $\eps^y_{1} = 0.1$.}
    \label{fig:risk}
\end{figure}
\begin{table}
    \centering
    \caption{Increase in total trajectory cost $J_{\mathrm{tot}}$ relative to the exact case (\textit{Exact}) }
    \begin{tabular}{lrrrr} \hline
         & Min~~ & Mean~ & Max~~ & Std Dev \\ \hline
        \textit{Proposed}: & $+1.63\,\%$ & $+2.29\,\%$ & $+4.85\,\%$ & $\pm0.36\,\%$\\
        \textit{Estimate}:  & $-1.55\,\%$ & $-0.44\,\%$ & $+0.01\,\%$ & $\pm0.15\,\%$\\ 
        \textit{OnlyIF}:  & $+5.04\,\%$ & $+8.53\,\%$ & $+12.88\,\%$ & $\pm1.26\,\%$\\ \hline
    \end{tabular}
    \label{tab:cost}
    
\end{table}

\section{Discussion} \label{sec:discuss}
The notion of consistent disturbance data in Section\;\ref{sec:consistency} enables the construction of causal data-driven predictors \eqref{eq:predictor_output_estimate}, \eqref{eq:predictor_extstate_estimate} based on an estimate \eqref{eq:distdata_estimate} of the unknown disturbance data that is consistent with the available input--output data and underlying system class. The evaluation of open-loop prediction errors in \cite[Sec.VI]{teutsch2024sampling} yields numerical evidence that consistent disturbance data estimates (and, thus, causal predictors) are favorable, as has also been reported by \cite{sader2025causality}. The presented analysis of prediction errors in Section\;\ref{sec:prederr} leverages the fact that data-driven predictors based on consistent disturbance data estimates are equivalent to corresponding model-based predictors. 
In fact, by taking disturbance data into account, Prop.~\ref{prop:consistency_multistep} bridges data-driven multi-step predictors commonly used in SPC with model-based predictors constructed from ARX parameters (see \eqref{eq:connection_sysparams_data_multistep_2}), thus unifying SPC and ARX-model-based MPC frameworks. The prediction error analysis of Section\;\ref{sec:prederr} is hence readily adoptable to a pure model-based setting.

Notably, as also common for related works on DPC, the proposed data-driven multi-step predictors \eqref{eq:predictor_output_estimate} and \eqref{eq:predictor_extstate_estimate} are less data efficient than the corresponding model-based predictors constructed from the single-step models. For a discussion on the issue and solution approaches, see \citep{alsalti2024sample}. For a general discussion on the differences between multi-step predictors and propagated single-step models, we refer to \cite{kohler2022state}.

We also remark that the over-approximation of the variance proxy in Lem.\;\ref{lem:prederror_data} is rather conservative due to the use of the Cauchy-Schwarz inequality and the upper-bound $\gamma_v$ from Asm.\;\ref{asm:strongstab}(c). The resulting constraint tightening approach is known in the literature for being conservative \citep{arcari2023stochastic}, and is also reflected in the numerical evaluation (see Sec.~\ref{sec:eval}). In principle, conservatism can be reduced by replacing $\gamma_v$ with $\norm{\col{\hankel{1}{\mathcal{X}_{T}},\hankel{1}{\mathcal{V}_{T}}}^{\dagger}\col{\bm{\xi}_{l|k},\bm{v}_{l|k}}}$, resulting in nonlinear (but convex) tightened constraints via tightening parameters that depend on the decision variable $\bm{v}_{\mathrm{f},k}$. A similar input-state dependent constraint tightening is derived by \cite{balim2024stochastic} and \cite{yin2023stochastic}. 
However, \cite{balim2024stochastic} require Gaussian disturbances and rely on a restrictive assumption \cite[Asm.~1]{balim2024stochastic} on the validity of parameter estimates for a finite amount of data. Furthermore, the constraint tightening in \cite[Lem.~8~\&~Cor.~10]{yin2023stochastic} relies on estimates of the system parameters, invalidating the guarantees if the employed estimates do not correspond to the true system parameters. In contrast, our approach considers sub-Gaussian disturbances, and we leverage the parameter bounds from Asm.~\ref{asm:strongstab}(a)--(b) to robustify our results against potential deviations between the employed estimates~\eqref{eq:sysparam_estimate} and the true parameters~\eqref{eq:connection_truesysparams_distdata}, see Lem.~\ref{lem:prederror_data} and~\ref{lem:prederror_add}. However, as common for set-based methods \citep{arcari2023stochastic}, the computational complexity increases with the dimension of $\mathbb{A}^{\mathrm{d}}$ from Asm.~\ref{asm:strongstab}(a), specifically for \eqref{eq:extstateerror_add_var} in Lem.~\ref{lem:prederror_add}. Additionally, the probabilistic prediction error bounds from Lem.~\ref{lem:prederror_data} grow with $n_{\theta} = n_y(n_{\xi}+n_u)$, i.e., the number of elements of the unknown system parameters $\bm{\Phi}$ and $\bm{\Psi}$, cf. \citep[Lem.~3]{kohler2022state}, \citep[Thm.~1]{balim2024stochastic}.

Lastly, we highlight that the considered ARX form of the system~\eqref{eq:system} can be restrictive since it fixes the disturbance model with respect to the system model. In a more general setting, the dynamics in \eqref{eq:system} might depend on autoregressive terms of the disturbance, i.e., $\bm{d}_k = \sum_{i=0}^{T_{\mathrm{p}}} \bm{\Upsilon}_{i}\bm{w}_{k-i}$ with unknown model parameters $\bm{\Upsilon}_{i}$ and disturbances $\bm{w}_{k-i}$. In this setting, $\bm{d}_k$ represents all the unknown disturbances that affect the system, and its distribution might be estimated from data \cite[Remark~2]{ou2025stochastic}. Specifically, when $\bm{w}$ is sub-Gaussian and $\bm{\Upsilon}_{i}$ is deterministic (or independent of $\bm{w}$) with known bounds as in Asm.~\ref{asm:strongstab}(a), \cite[Thm.~1]{ao2025stochastic} can be applied to obtain an over-approximation of the variance proxy $\bm{\Sigma}_d$ of $\bm{d}_k$, similar to Lem.~\ref{lem:prederror_add}. However, $\bm{d}_k$ might not be independent anymore (cf. Asm.~\ref{asm:disturb}). Considering the cases of correlated or conditionally independent disturbances, \cite[Thm.~1]{ao2025stochastic}could potentially be exploited toward generalizing the proposed approach; a detailed analysis is out of the scope of this work and thus left open for future research.


\section{Conclusion} \label{sec:conclusion}
In this work, we have proposed an output-feedback stochastic data-driven predictive controller for linear systems subject to sub-Gaussian additive disturbances. By leveraging disturbance data estimates that are consistent with the available input--output data and system class, causality is enforced over the data-driven predictions. The causal structure of the identified subspace predictors and their connection to corresponding model-based predictors enable a straightforward analysis of statistical properties of the prediction errors. Deterministic reformulations of chance constraints are obtained by leveraging concentration inequalities for (norms of) sub-Gaussian random variables. To guarantee recursive feasibility and chance constraint satisfaction in closed-loop under the proposed control scheme, we combine classical and indirect feedback formulations from the literature. Despite conservative constraint tightening, numerical evidence shows that the proposed stochastic data-driven predictive control scheme can achieve efficient closed-loop performance.

\def\bibfont{\small}
\bibliographystyle{elsarticle-num-names} 
\bibliography{references.bib}                                              

\appendix
\section{Strongly Stabilizing Feedback} \label{app:strongstab}
Here, we discuss how to obtain a $(c_\xi,\gamma_\xi)$-strongly stabilizing feedback gain $\bm{K}$ for the extended-state dynamics~\eqref{eq:system_arx_nonminimal}, see Asm.~\ref{asm:strongstab}(b). The following result generalizes \cite[Prop.~1]{kerz2024safe} to the considered setting.
\begin{proposition} \label{prop:strongstab}
    Let $(\bm{A}_i,\bm{B}_i)$, $i\in\mathbb{N}_{1}^{N_{\text{v}}}$, be the $N_{\text{v}}$ vertices of the set of system matrices $\mathbb{A}^{\mathrm{d}}$ from Asm~\ref{asm:strongstab}(a). Further, for a fixed $\gamma_\xi\in(0,1)$ and matrices $\bm{M}_1$ and $\bm{M}_2$ of suitable dimensions, let $(c^*, \bm{Z}^*,\bm{Y}^*)$ be the solution of
    \begin{subequations}
        \begin{align}
            &\underset{c>0,~\bm{Z} = \bm{Z}^{\top} \in\mathbb{R}^{n_{\xi}\times n_{\xi}},~\bm{Y} \in\mathbb{R}^{n_u\times n_{\xi}}}{\mathrm{minimize}} ~~ c \label{eq:strongstab_cost}\\
            &\mathrm{s.t.} ~~  c\bm{I}_{n_{\xi}}\succeq \bm{Z} \succeq\bm{M}_2\bm{M}_2^{\top} \label{eq:strongstab_bounds}\\
            &\phantom{\mathrm{s.t.}} ~~ \mat{\gamma_\xi \bm{Z} & \bm{A}_i \bm{Z} + \bm{B}_i \bm{Y} \\ (\bm{A}_i \bm{Z} + \bm{B}_i \bm{Y})^{\top} & \gamma_\xi \bm{Z}} \succ \bm{0} ~~\forall i\in \mathbb{N}_{1}^{N_{\text{v}}}. \label{eq:strongstab_lmi}
        \end{align}
    \end{subequations}
    Then, for all $\bm{A}_{\mathrm{cl}} = \bm{A} + \bm{B}\bm{K}$ with $(\bm{A},\bm{B})\in\mathbb{A}^{\mathrm{d}}$ and feedback gain $\bm{K} = \bm{Y}^*(\bm{Z}^*)^{-1}$, it holds that $\norm{\bm{M}_1\bm{A}_{\mathrm{cl}}^k\bm{M}_2} \le c_M \gamma_\xi^k$ for $c_M = \sqrt{\lambda_{\max}(\bm{M}_1\bm{Z}^*\bm{M}_1^{\top})}$ and all $k\ge 0$, where $\lambda_{\max}(\cdot)$ yields the largest eigenvalue of a matrix. 
    \hfill \small{$\square$}
\end{proposition}
\begin{proof} 
    If condition \eqref{eq:strongstab_lmi} is satisfied, then, by the Schur complement and due to convexity of $\mathbb{A}^{\mathrm{d}}$, it holds that
    \begin{equation*}
        \gamma_\xi \bm{Z}^* - (\bm{A} \bm{Z}^* + \bm{B} \bm{Y}^*) (\gamma_\xi \bm{Z}^*)^{-1} (\bm{A} \bm{Z}^* + \bm{B} \bm{Y}^*)^{\top} \succ \bm{0}
    \end{equation*}
    for all $(\bm{A},\bm{B})\in\mathbb{A}^{\mathrm{d}}$ . Factorizing $\bm{Z}^*$, dividing by $\gamma_\xi$, and defining $\bm{K} \coloneqq \bm{Y}^*(\bm{Z}^*)^{-1}$ and $\check{\bm{A}}_{\mathrm{cl}} \coloneqq \gamma_\xi^{-1}\bm{A}_{\mathrm{cl}}$, we obtain
    \begin{equation}
        \bm{Z}^* - \check{\bm{A}}_{\mathrm{cl}} \bm{Z}^* \check{\bm{A}}_{\mathrm{cl}}^{\top} \succ \bm{0}. \label{eq:strongstab_lyap}
    \end{equation}
    Now, investigating the term $\norm{\bm{M}_1\check{\bm{A}}_{\mathrm{cl}}^k\bm{M}_2}$, it follows that
    \begin{align}
        \norm{\bm{M}_1\check{\bm{A}}_{\mathrm{cl}}^k\bm{M}_2} &= \sqrt{\lambda_{\max}(\bm{M}_1\check{\bm{A}}_{\mathrm{cl}}^k\bm{M}_2\bm{M}_2^{\top}(\check{\bm{A}}_{\mathrm{cl}}^k)^{\top}\bm{M}_1^{\top})} \notag\\
        &\overset{\eqref{eq:strongstab_bounds}}{\le} \sqrt{\lambda_{\max}(\bm{M}_1\check{\bm{A}}_{\mathrm{cl}}^k\bm{Z}^*(\check{\bm{A}}_{\mathrm{cl}}^k)^{\top}\bm{M}_1^{\top})} \notag\\
        &\overset{\eqref{eq:strongstab_lyap}}{\le} \sqrt{\lambda_{\max}(\bm{M}_1\bm{Z}^*\bm{M}_1^{\top})} =:c_M \label{eq:strongstab_finalineq},
    \end{align}
    leveraging $\check{\bm{A}}_{\mathrm{cl}}^k = \check{\bm{A}}_{\mathrm{cl}}^{k-1}\check{\bm{A}}_{\mathrm{cl}}$ and repeated application of~\eqref{eq:strongstab_lyap}. Lastly, multiplying both sides of inequality \eqref{eq:strongstab_finalineq} with $\gamma_\xi^k$, we obtain the assertion $\norm{\bm{M}_1{\bm{A}}_{\mathrm{cl}}^k\bm{M}_2} \le c_M \gamma_\xi^k$.
    \hfill \small{$\blacksquare$}
\end{proof}
For $\bm{M}_1 = \bm{M_2} = \bm{I}_{n_{\xi}}$, we obtain $\norm{\bm{A}_{\mathrm{cl}}^k} \le c_\xi \gamma_\xi^k$ as in Asm.~\ref{asm:strongstab}(b), i.e., the gain $\bm{K}$ is $(c_\xi,\gamma_\xi)$-strongly stabilizing with $c_\xi = \sqrt{c^*}$. Thus, minimizing $c$ in \eqref{eq:strongstab_cost} ensures that the eigenvalues of $\bm{Z}^*$ are as small as possible. Furthermore, by choosing $\bm{M}_2 = \bm{E}_y^{\top}\bm{\Sigma}_d^{\nicefrac{1}{2}}$ and $\bm{M}_1 = [\bm{G}_y]_i\bm{E}_y$ ($\bm{M}_1 = [\bm{G}_u]_j\bm{E}_u$) for $i \in \mathbb{N}_1^{r_y}$ ($j \in \mathbb{N}_1^{r_u}$), we obtain tight parameters $c_{i}^{y} = c_M$  ($c_{j}^{u} = c_M$) for Lem.~\ref{lem:prederror_data}.

\section{Proof of Proposition\;\ref{prop:consistency_multistep}} \label{app:proof_consistency}
We first prove \eqref{eq:consistency_multistep} following the arguments from \cite[Prop.~2]{pan2021stochastic}. 
Full row-rank of $\bm{S}_{2}$ follows from Asm.\;\ref{asm:trajData}(a)--(b), cf. \citep[Lemma~1]{de2019formulas}. Thus, using $\bm{S}_{2}\bm{S}_{2}^{\dagger}\bm{S}_{2} = \bm{S}_{2}$, we write \eqref{eq:dynamics_data_multistep} as
\begin{align}
    \hankel{T_{\mathrm{f}}}{\mathcal{Y}_{T}} &= \mat{\mathcal{O}_{T_{\mathrm{f}}} & \mathcal{T}_{T_{\mathrm{f}}}^{v} } \bm{S}_{2} + \mathcal{T}_{T_{\mathrm{f}}}^{d}\hankel{T_{\mathrm{f}}}{\mathcal{D}_{T}} \label{eq:system_data_multistep_helper}\\
    &= \mat{\mathcal{O}_{T_{\mathrm{f}}} & \mathcal{T}_{T_{\mathrm{f}}}^{v} } \bm{S}_{2}\bm{S}_{2}^{\dagger}\bm{S}_{2} + \mathcal{T}_{T_{\mathrm{f}}}^{d}\hankel{T_{\mathrm{f}}}{\mathcal{D}_{T}}. \notag
\end{align}
By substituting $\mat{\mathcal{O}_{T_{\mathrm{f}}} & \mathcal{T}_{T_{\mathrm{f}}}^{v} } \bm{S}_{2}$ with $\hankel{T_{\mathrm{f}}}{\mathcal{Y}_{T}} - \mathcal{T}_{T_{\mathrm{f}}}^{d}\hankel{T_{\mathrm{f}}}{\mathcal{D}_{T}}$, we obtain the following relation equivalent to \eqref{eq:consistency_multistep}:
\begin{equation*}
    \hankel{T_{\mathrm{f}}}{\mathcal{Y}_{T}} = \left(\hankel{T_{\mathrm{f}}}{\mathcal{Y}_{T}} - \mathcal{T}_{T_{\mathrm{f}}}^{d}\hankel{T_{\mathrm{f}}}{\mathcal{D}_{T}}\right)\bm{S}_{2}^{\dagger}\bm{S}_{2} + \mathcal{T}_{T_{\mathrm{f}}}^{d}\hankel{T_{\mathrm{f}}}{\mathcal{D}_{T}}.
\end{equation*}

Let us now prove \eqref{eq:connection_sysparams_data_multistep}. From satisfaction of the $1$-step consistency constraint~\eqref{eq:consistency}, the data equation~\eqref{eq:dynamics_data} holds with parameters from~\eqref{eq:connection_sysparams_data}. By recursively extracting the corresponding $\tilde{T}$ columns and applying \eqref{eq:dynamics_data} considering the PE conditions from Asm.\;\ref{asm:trajData}(a)--(b), the $T_{\mathrm{f}}$-step data equation~\eqref{eq:dynamics_data_multistep} follows with corresponding parameters $\hat{\mathcal{O}}_{T_{\mathrm{f}}}$, $\hat{\mathcal{T}}_{T_{\mathrm{f}}}^{v}$, $\hat{\mathcal{T}}_{T_{\mathrm{f}}}^{d}$ (cf. \cite{de2019formulas}). Then, as shown before, \eqref{eq:dynamics_data_multistep} implies that $(\hat{\mathcal{D}}_{T},\,\hat{\mathcal{T}}_{T_{\mathrm{f}}}^{d})$ satisfy  the $T_{\mathrm{f}}$-step consistency constraint \eqref{eq:consistency_multistep}.
Now, since $\bm{\Pi}^v_{T_{\mathrm{f}}}$ in \eqref{eq:consistency_multistep} is the orthogonal projector onto the null-space of $\bm{S}_{2}$, we have that $(\hat{\mathcal{D}}_{T},\,\hat{\mathcal{T}}_{T_{\mathrm{f}}}^{d})$ satisfying \eqref{eq:consistency_multistep} implies $\hankel{T_{\mathrm{f}}}{\mathcal{Y}_{T}} - \hat{\mathcal{T}}_{T_{\mathrm{f}}}^{d}\hankel{T_{\mathrm{f}}}{\hat{\mathcal{D}}_{T}}$ lying in the image space of $\bm{S}_{2}$, i.e., there exists a matrix $\bm{\Theta} \in \mathbb{R}^{n_y T_{\mathrm{f}} \times (n_{\xi} + (n_u + n_y)T_{\mathrm{f}})}$ such that
\begin{equation*}
    \hankel{T_{\mathrm{f}}}{\mathcal{Y}_{T}} - \hat{\mathcal{T}}_{T_{\mathrm{f}}}^{d}\hankel{T_{\mathrm{f}}}{\hat{\mathcal{D}}_{T}} = \bm{\Theta} \bm{S}_{2}.
\end{equation*}
By splitting $\bm{\Theta}$ into $\mat{\hat{\mathcal{O}}_{T_{\mathrm{f}}} & \hat{\mathcal{T}}_{T_{\mathrm{f}}}^{v}}$, we obtain relation \eqref{eq:connection_sysparams_data_multistep}.

Lastly, we prove equivalence of \eqref{eq:connection_sysparams_data_multistep} and \eqref{eq:connection_sysparams_data_multistep_2} under satisfaction of the $T_{\mathrm{f}}$-step consistency constraint~\eqref{eq:consistency_multistep}. Using the properties of the pseudo-inverse, we can write
\begin{equation*}
    \hankel{T_{\mathrm{f}}}{\mathcal{Y}_{T}} \col{\bm{S}_{2},\hankel{T_{\mathrm{f}}}{\hat{\mathcal{D}}_{T}}}^{\dagger} = \mat{\bm{\Theta}_1 & \bm{\Theta}_2}
\end{equation*}
with matrices $\bm{\Theta}_2 \coloneqq \hankel{T_{\mathrm{f}}}{\mathcal{Y}_{T}}\big(\hankel{T_{\mathrm{f}}}{\hat{\mathcal{D}}_{T}}\bm{\Pi}^{v}_{T_{\mathrm{f}}}\big)^{\dagger}$, $\bm{\Theta}_1 \coloneqq \hankel{T_{\mathrm{f}}}{\mathcal{Y}_{T}} \left(\bm{S}_{2}\bm{\Pi}^{d}_{T_{\mathrm{f}}}\right)^{\dagger}$, and $\bm{\Pi}^{d}_{T_{\mathrm{f}}}\coloneqq \bm{I}_{\tilde{T}} - \hankel{T_{\mathrm{f}}}{\hat{\mathcal{D}}_{T}}^{\dagger}\hankel{T_{\mathrm{f}}}{\hat{\mathcal{D}}_{T}}$.
To prove the claim, we need to show that $(\hat{\mathcal{D}}_{T},\,\hat{\mathcal{T}}_{T_{\mathrm{f}}}^{d})$ satisfying~\eqref{eq:consistency_multistep} implies (i) $\bm{\Theta}_2 = \hat{\mathcal{T}}_{T_{\mathrm{f}}}^{d}$ and (ii) $\bm{\Theta}_1 = \mat{\hat{\mathcal{O}}_{T_{\mathrm{f}}} & \hat{\mathcal{T}}_{T_{\mathrm{f}}}^{v}}$. Equivalence of \eqref{eq:connection_sysparams_data_multistep} and \eqref{eq:connection_sysparams_data_multistep_2} is then implied, as both relations are derived from the same constraint~\eqref{eq:consistency_multistep}.

As the projection matrix $\bm{\Pi}^{v}_{T_{\mathrm{f}}}$ is idempotent, we can equivalently write $\bm{\Theta}_2 = \hankel{T_{\mathrm{f}}}{\mathcal{Y}_{T}}\bm{\Pi}^{v}_{T_{\mathrm{f}}}\left(\hankel{T_{\mathrm{f}}}{\hat{\mathcal{D}_{T}}}\bm{\Pi}^{v}_{T_{\mathrm{f}}}\right)^{\dagger}$.
If $(\hat{\mathcal{D}}_{T},\,\hat{\mathcal{T}}_{T_{\mathrm{f}}}^{d})$ satisfy~\eqref{eq:consistency_multistep}, we have that $\hankel{T_{\mathrm{f}}}{\mathcal{Y}_{T}}\bm{\Pi}^{v}_{T_{\mathrm{f}}} = \hat{\mathcal{T}}_{T_{\mathrm{f}}}^{d}\hankel{T_{\mathrm{f}}}{\hat{\mathcal{D}}_{T}}\bm{\Pi}^{v}_{T_{\mathrm{f}}}$. Thus, claim (i) follows from $\bm{\Theta}_2 = \hat{\mathcal{T}}_{T_{\mathrm{f}}}^{d}\hankel{T_{\mathrm{f}}}{\hat{\mathcal{D}}_{T}}\bm{\Pi}^{v}_{T_{\mathrm{f}}}\left(\hankel{T_{\mathrm{f}}}{\mathcal{D}_{T}}\bm{\Pi}^{v}_{T_{\mathrm{f}}}\right)^{\dagger} = \hat{\mathcal{T}}_{T_{\mathrm{f}}}^{d}$.

Further, as the projection matrix $\bm{\Pi}^{d}_{T_{\mathrm{f}}}$ is idempotent, we can write $\bm{\Theta}_1 = \hankel{T_{\mathrm{f}}}{\mathcal{Y}_{T}}\bm{\Pi}^{d}_{T_{\mathrm{f}}} \left(\bm{S}_{2}\bm{\Pi}^{d}_{T_{\mathrm{f}}}\right)^{\dagger}$. As before, $(\hat{\mathcal{D}}_{T},\,\hat{\mathcal{T}}_{T_{\mathrm{f}}}^{d})$ satisfying~\eqref{eq:consistency_multistep} means that relation~\eqref{eq:system_data_multistep_helper} holds. 
Thus, claim (ii) follows from
\begin{align*}
    \bm{\Theta}_1 & =\hankel{T_{\mathrm{f}}}{\mathcal{Y}_{T}}\bm{\Pi}^{d}_{T_{\mathrm{f}}}\hspace{-2pt} \left(\bm{S}_{2}\bm{\Pi}^{d}_{T_{\mathrm{f}}}\right)^{\dagger}\hspace{-2pt} = \mat{\hat{\mathcal{O}}_{T_{\mathrm{f}}} & \hat{\mathcal{T}}_{T_{\mathrm{f}}}^{v} } \bm{S}_{2}\bm{\Pi}^{d}_{T_{\mathrm{f}}} \hspace{-2pt}\left(\bm{S}_{2}\bm{\Pi}^{d}_{T_{\mathrm{f}}}\right)^{\dagger}   \\
    &~~~~~~ + \hat{\mathcal{T}}_{T_{\mathrm{f}}}^{d}\hankel{T_{\mathrm{f}}}{\hat{\mathcal{D}}_{T}}\bm{\Pi}^{d}_{T_{\mathrm{f}}} \left(\bm{S}_{2}\bm{\Pi}^{d}_{T_{\mathrm{f}}}\right)^{\dagger} = \mat{\hat{\mathcal{O}}_{T_{\mathrm{f}}} & \hat{\mathcal{T}}_{T_{\mathrm{f}}}^{v} },
\end{align*}
since $\hankel{T_{\mathrm{f}}}{\hat{\mathcal{D}}_{T}}\bm{\Pi}^{d}_{T_{\mathrm{f}}} = \bm{0}$, completing the proof. 
\section{Proof of Lemma~\ref{lem:prederror_data}} \label{app:proof_prederror_data}
Consider the abbreviation $\bm{S}_1 = \col{\hankel{1}{\mathcal{X}_{T}},\hankel{1}{\mathcal{V}_{T}}}$ as in Prop.~\ref{prop:consistency}, and let us define the extended-state error $\bm{e}^{\xi,\text{data}}_{l+1|k}$ for predicted time step $l+1$ such that $\bm{e}^{y,\text{data}}_{l|k} = \bm{E}_y \bm{e}^{\xi,\text{data}}_{l+1|k}$ and $\bm{e}^{u,\text{data}}_{l|k} = \bm{E}_u \bm{e}^{\xi,\text{data}}_{l+1|k}$, respectively. Then, using $\bm{\Phi}_{\mathrm{cl}} = \bm{E}_y\bm{A}_{\mathrm{cl}}$ and \eqref{eq:prederror_output}, we have
\begin{align*}
    \bm{e}^{\xi,\text{data}}_{l+1|k} = (\bm{A}_{\mathrm{cl}}^{l+1} - \hat{\bm{A}}_{\mathrm{cl}}^{l+1}) \bm{\xi}_{0|k} + \textstyle\sum_{\iota=0}^l(\bm{A}_{\mathrm{cl}}^{l-\iota}\bm{B} - \hat{\bm{A}}_{\mathrm{cl}}^{l-\iota}\hat{\bm{B}})\bm{v}_{\iota|k}
\end{align*}
with the true system matrices $\bm{A}_{\mathrm{cl}} = \bm{A} + \bm{B}\bm{K}$ and $\bm{B}$ of system~\eqref{eq:system_arx_nonminimal} and the estimates $\hat{\bm{A}}_{\mathrm{cl}}$, $\hat{\bm{B}}$ corresponding to the disturbance data estimate $\hat{\mathcal{D}}_{T}$ from \eqref{eq:distdata_estimate} via the relation \eqref{eq:connection_sysparams_data} and \eqref{eq:system_arx_nonminimal_helper}. Now, note that for all $\iota \in \mathbb{N}_0^{l-1}$, we have 
\begin{align*}
    &(\bm{A}_{\mathrm{cl}}^{l-\iota+1} - \hat{\bm{A}}_{\mathrm{cl}}^{l-\iota+1}) \bm{\xi}_{\iota|k} + (\bm{A}_{\mathrm{cl}}^{l-\iota}\bm{B} - \hat{\bm{A}}_{\mathrm{cl}}^{l-\iota}\hat{\bm{B}})\bm{v}_{\iota|k} \\
    &~= \bm{A}_{\mathrm{cl}}^{l-\iota} \mat{\bm{A}_{\mathrm{cl}} & \bm{B}} \mat{\bm{\xi}_{\iota|k}\\\bm{v}_{\iota|k}} - \hat{\bm{A}}_{\mathrm{cl}}^{l-\iota} \mat{\hat{\bm{A}}_{\mathrm{cl}} & \hat{\bm{B}}} \mat{\bm{\xi}_{\iota|k}\\\bm{v}_{\iota|k}}\\
    &~= \bm{A}_{\mathrm{cl}}^{l-\iota} \mat{\bm{A}_{\mathrm{cl}} - \hat{\bm{A}}_{\mathrm{cl}} & \bm{B} - \hat{\bm{B}}} \mat{\bm{\xi}_{\iota|k}\\\bm{v}_{\iota|k}} + (\bm{A}_{\mathrm{cl}}^{l-\iota} -\hat{\bm{A}}_{\mathrm{cl}}^{l-\iota}) \bm{\xi}_{\iota+1|k}
\end{align*}
with $\bm{\xi}_{\iota+1|k} = \mat{\hat{\bm{A}}_{\mathrm{cl}} & \hat{\bm{B}}} \col{\bm{\xi}_{\iota|k},\bm{v}_{\iota|k}}$ via the causal structure of \eqref{eq:predictor_extstate_estimate}. From this reformulation and \eqref{eq:sysparam_error}, we obtain     
\begin{align}
    \bm{e}^{\xi,\text{data}}_{l+1|k} &= -\sum_{\iota=0}^{l} \bm{A}_{\mathrm{cl}}^{l-\iota}\bm{E}_y^{\top}\hankel{1}{{\mathcal{D}}_{T}} \bm{S}_1^{\dagger}\vc{\bm{\xi}_{\iota|k}\\\bm{v}_{\iota|k}},\notag
\end{align}
since $\mat{\bm{A}_{\mathrm{cl}} - \hat{\bm{A}}_{\mathrm{cl}} & \bm{B} - \hat{\bm{B}}} = \bm{E}_y^\top\mat{\bm{\Phi}_{\mathrm{cl}} - \hat{\bm{\Phi}}_{\mathrm{cl}} & \bm{\Psi} - \hat{\bm{\Psi}}}$ via \eqref{eq:system_arx_nonminimal_helper}. Leveraging $[\bm{G}_y]_i\bm{y}_{l|k} = [\bm{G}_y]_i\bm{E}_y\bm{\xi}_{l+1|k}$ and $[\bm{G}_u]_j\bm{u}_{l|k} = [\bm{G}_u]_j\bm{E}_u\bm{\xi}_{l+1|k}$ with $i\in\mathbb{N}_1^{r_y}$ and $j\in\mathbb{N}_1^{r_u}$, we further obtain
\begin{align*}
    [\bm{G}_y]_i\bm{e}^{y,\text{data}}_{l|k} = -\sum_{\iota=0}^{l} [\bm{G}_y]_i\bm{E}_y\bm{A}_{\mathrm{cl}}^{l-\iota}\bm{E}_y^{\top}\hankel{1}{{\mathcal{D}}_{T}} \bm{S}_1^{\dagger}\vc{\bm{\xi}_{\iota|k}\\\bm{v}_{\iota|k}},\\
    [\bm{G}_u]_j\bm{e}^{u,\text{data}}_{l|k} = -\sum_{\iota=0}^{l} [\bm{G}_u]_j\bm{E}_u\bm{A}_{\mathrm{cl}}^{l-\iota}\bm{E}_y^{\top}\hankel{1}{{\mathcal{D}}_{T}} \bm{S}_1^{\dagger}\vc{\bm{\xi}_{\iota|k}\\\bm{v}_{\iota|k}},
\end{align*}
for the output and input error terms, respectively.

Since $\bm{S}_1^{\dagger}$ has full column-rank due to Asm.~\ref{asm:trajData}(b), we can write $\bm{S}_1^{\dagger} = \bm{S}_{\text{l}} \bm{S}_{\text{r}}$ with $\bm{S}_{\text{l}} \in \mathbb{R}^{T\times (n_{\xi}+n_u)}$ and $\bm{S}_{\text{r}} \in \mathbb{R}^{(n_{\xi}+n_u)\times (n_{\xi}+n_u)}$ such that $\bm{S}_{\text{l}}^{\top}\bm{S}_{\text{l}} = \bm{I}_{n_{\xi}+n_u}$ and $\norm{\bm{S}_{\text{r}} \col{\bm{\xi},\bm{v}}} = \lVert{\bm{S}_{1}^{\dagger} \col{\bm{\xi},\bm{v}}}\rVert$ for all $\col{\bm{\xi},\bm{v}}\in \mathbb{R}^{n_{\xi}+n_u}$.
Moreover, we define $\ul{{\bm{d}}}_{\theta}$ as the column-wise vectorization of $(\bm{\Sigma}_d^{\nicefrac{1}{2}})^{-1}\hankel{1}{{\mathcal{D}}_{T}}\bm{S}_{\text{l}}$, being zero-mean sub-Gaussian with variance proxy $(\bm{S}_{\text{l}}^{\top} \otimes (\bm{\Sigma}_d^{\nicefrac{1}{2}})^{-1})(\bm{I}_T \otimes \bm{\Sigma}_d)(\bm{S}_{\text{l}}^{\top} \otimes (\bm{\Sigma}_d^{\nicefrac{1}{2}})^{-1})^{\top} = \bm{I}_{n_{\theta}}$ via Asm.~\ref{asm:disturb} and \citep[Thm.~1]{ao2025stochastic} with $n_{\theta} = n_y(n_{\xi}+n_u)$. Then, we can write
\begin{align*}
    &[\bm{G}_y]_i\bm{e}^{y,\text{data}}_{l|k} = \\
    &\hspace{5mm}-\left(\sum_{\iota=0}^{l} \left(\bm{S}_{\text{r}}\vc{\bm{\xi}_{\iota|k}\\\bm{v}_{\iota|k}}\right)^{\top} \otimes [\bm{G}_y]_i\bm{E}_y\bm{A}_{\mathrm{cl}}^{l-\iota}\bm{E}_y^{\top}\bm{\Sigma}_d^{\nicefrac{1}{2}}\right) \ul{{\bm{d}}}_{\theta}, \notag\\
    &[\bm{G}_u]_j\bm{e}^{u,\text{data}}_{l|k} = \\
    &\hspace{5mm}-\left(\sum_{\iota=0}^{l} \left(\bm{S}_{\text{r}}\vc{\bm{\xi}_{\iota|k}\\\bm{v}_{\iota|k}}\right)^{\top} \otimes [\bm{G}_u]_j\bm{E}_u\bm{A}_{\mathrm{cl}}^{l-\iota}\bm{E}_y^{\top}\bm{\Sigma}_d^{\nicefrac{1}{2}}\right) \ul{{\bm{d}}}_{\theta}.\notag
\end{align*}  
By applying the Cauchy-Schwarz inequality and leveraging condition~\eqref{eq:inputstatebound_predicted} and the constants  $\gamma_v$, $c_{i}^{y}$, $c_{j}^{u}$, $\gamma_\xi$, we obtain 
\begin{subequations}
    \begin{align*}
        &[\bm{G}_y]_i\bm{e}^{y,\text{data}}_{l|k} \le  \big(\textstyle\sum_{\iota=0}^{l} \gamma_v c_{i}^{y}\gamma_{\xi}^{\iota}\big)  \norm{\ul{{\bm{d}}}_{\theta}},\\
        &[\bm{G}_u]_j\bm{e}^{u,\text{data}}_{l|k} \le \big(\textstyle\sum_{\iota=0}^{l} \gamma_v c_{j}^{u}\gamma_{\xi}^{\iota}\big) \norm{\ul{{\bm{d}}}_{\theta}}.
    \end{align*}  
\end{subequations}
The assertion then follows by applying \citep[Thm.~2]{ao2025stochastic}, yielding $\PrBig{\norm{\ul{{\bm{d}}}_{\theta}} \le \sqrt{n_{\theta}(1+f^{-1}(\eps^{-\frac{2}{n_{\theta}}}))}} \ge 1-\eps$ for any $\eps\in(0,1)$, with $f(\delta) = \exp{\delta}/(1+\delta)$.

\section{Proof of Lemma~\ref{lem:prederror_add}} \label{app:proof_prederror_add}
Note that, by Asm.\;\ref{asm:disturb} and \cite[Thm.~1]{ao2025stochastic}, $\bm{e}^{y,\text{add}}_{\mathrm{f},k} = \mathcal{T}_{T_{\mathrm{f}}}^{d} \bm{d}_{\mathrm{f},k}$ follows a zero-mean sub-Gaussian distribution with variance proxy $\mathcal{T}_{T_{\mathrm{f}}}^{d} (\bm{I}_{T_{\mathrm{f}}} \otimes \bm{\Sigma}_d) (\mathcal{T}_{T_{\mathrm{f}}}^{d})^{\top}$. Similar to \ref{app:proof_prederror_data}, let us define the corresponding extended-state error $\bm{e}^{\xi,\text{add}}_{l+1|k}$ for predicted time step $l+1$ such that $\bm{e}^{y,\text{add}}_{l|k} = \bm{E}_y \bm{e}^{\xi,\text{add}}_{l+1|k}$ and $\bm{e}^{u,\text{add}}_{l|k} = \bm{E}_u \bm{e}^{\xi,\text{add}}_{l+1|k}$, $l\in\mathbb{N}_0^{T_{\mathrm{f}}-1}$. Exploiting $\bm{\Phi}_{\mathrm{cl}} = \bm{E}_y\bm{A}_{\mathrm{cl}}$ and Asm.\;\ref{asm:strongstab}(a), we can apply the arguments presented in \cite[Sec.~III.B]{arcari2023stochastic} to conclude that the extended-state error $\bm{e}^{\xi,\text{add}}_{l+1|k}$ follows a zero-mean sub-Gaussian distribution  with the variance proxy over-approximation $\ol{\bm{\Sigma}}^{\xi,\text{add}}_{l+1}$ satisfying \eqref{eq:extstateerror_add_var}.
The assertion then follows from \cite[Thm.~1]{ao2025stochastic} by leveraging $[\bm{G}_y]_i\bm{y}_{l|k} = [\bm{G}_y]_i\bm{E}_y\bm{\xi}_{l+1|k}$ and $[\bm{G}_u]_j\bm{u}_{l|k} = [\bm{G}_u]_j\bm{E}_u\bm{\xi}_{l+1|k}$ with $i\in\mathbb{N}_1^{r_y}$ and $j\in\mathbb{N}_1^{r_u}$.

\section{Proof of Lemma~\ref{lem:constight}}\label{app:proof_constight}
Consider any $l \in \mathbb{N}_0^{T_{\mathrm{f}}-1}$, $i\in\mathbb{N}_1^{r_y}$, and $j\in\mathbb{N}_1^{r_u}$. Applying \cite[Lem.~2]{ao2025stochastic} and Lem.~\ref{lem:prederror_add}, we obtain $\PrBig{\left[\bm{G}_y\right]_i \bm{e}^{y,\text{add}}_{l|k} \le \sqrt{2\ln(1/\check{\eps}^y)\ol{\bm{\Sigma}}^{y,\text{add}}_{l,i}}} \ge 1-\check{\eps}^y$ and $\PrBig{\left[\bm{G}_u\right]_j \bm{e}^{u,\text{add}}_{l|k} \le \sqrt{2\ln(1/\check{\eps}^u)\ol{\bm{\Sigma}}^{u,\text{add}}_{l,j}}} \ge 1-\check{\eps}^u$ for any $\check{\eps}^y, \check{\eps}^u \in (0,1)$. Further, from Lem.~\ref{lem:prederror_data} and Boole's inequality with $\tilde{\eps}^y = \check{\eps}^y = \eps_i^y/2$ and $\tilde{\eps}^u = \check{\eps}^u = \eps_j^u/2$, we obtain $\PrBig{\left[\bm{G}_y\right]_i \bm{e}^y_{l|k} \le \left[\bm{\eta}_l^y\right]_i} \ge 1-\eps^y_{i}$ and $\PrBig{\left[\bm{G}_u\right]_j \bm{e}^u_{l|k} \le \left[\bm{\eta}_l^u\right]_j} \ge 1-\eps^u_{j}$ with $\bm{\eta}_l^y$ and $\bm{\eta}_l^u$ from \eqref{eq:tightening_params}, and with the total prediction errors $\bm{e}^y_{l|k} = \bm{e}^{y,\text{data}}_{l|k} + \bm{e}^{y,\text{add}}_{l|k}$ and $\bm{e}^u_{l|k} = \bm{e}^{u,\text{data}}_{l|k} + \bm{e}^{u,\text{add}}_{l|k}$. The assertion then follows from $\bm{y}_{k+l} = \bm{y}_{l|k} + \bm{e}^y_{l|k}$, $\bm{u}_{k+l} = \bm{u}_{l|k} + \bm{e}^u_{l|k}$, and \eqref{eq:tightcons}.

\section{Proof of Theorem~\ref{thm:properties}} \label{app:proof_properties}

Recursive feasibility follows from standard arguments presented in the indirect feedback literature~\cite[Thm.~1]{hewing2020recursively}: Feasibility for $k=0$ and $\bm{\xi}_{0|0} = \bm{\xi}_0 \in \mathbb{X}$ is assumed, with $\mathbb{X}$ derived from the nominal input and output constraints as outlined in Asm.~\ref{asm:strongstab}(c). Feasibility for $k>0$ is proven using the indirect feedback initialization $\bm{\xi}_{0|k} = \bm{\xi}_{1|k-1}$: Let $\bm{v}^*_{\mathrm{f},k} = \col{\bm{v}^*_{0|k},\dots,\bm{v}^*_{T_{\mathrm{f}}-1|k}}$ be the feasible solution from time step $k\ge0$ with given $\bm{\xi}^*_{0|k} = \tilde{\bm{\xi}}_{k}$, with corresponding inputs $\bm{u}^*_{\mathrm{f},k} = \col{\bm{u}^*_{0|k},\dots,\bm{u}^*_{T_{\mathrm{f}}-1|k}}$, outputs $\bm{y}^*_{\mathrm{f},k} = \col{\bm{y}^*_{0|k},\dots,\bm{y}^*_{T_{\mathrm{f}}-1|k}}$, and extended states $\bm{\xi}^*_{\mathrm{f},k} = \col{\bm{\xi}^*_{1|k},\dots,\bm{\xi}^*_{T_{\mathrm{f}}|k}}$. For step $k+1$, we do not reset $\tilde{k}$, and we define $\bm{\xi}_{0|k+1} = \bm{\xi}^*_{1|k}$ and the candidates $\tilde{\bm{v}}_{\mathrm{f},k+1} = \col{\bm{v}^*_{1|k},\dots,\bm{v}^*_{T_{\mathrm{f}}-1|k},\bm{0}}$ and $\tilde{\bm{\xi}}_{\mathrm{f},k+1} = \col{\bm{\xi}^*_{2|k},\dots,\bm{\xi}^*_{T_{\mathrm{f}}|k},\hat{\bm{A}}_{\mathrm{cl}}\bm{\xi}^*_{T_{\mathrm{f}}|k}}$. Note that $\bm{\xi}^*_{0|k}\in \mathbb{X} \implies \bm{\xi}^*_{1|k} \in \mathbb{X}$, since $\bm{\xi}^*_{1|k}$ is constructed from $\bm{\xi}^*_{0|k}$ and predicted inputs and outputs that satisfy the tightened nominal constraints. Accordingly, we obtain the candidates $\tilde{\bm{u}}_{\mathrm{f},k+1} = \col{\bm{u}^*_{1|k},\dots,\bm{u}^*_{T_{\mathrm{f}}-1|k},\bm{E}_u\hat{\bm{A}}_{\mathrm{cl}}\bm{\xi}^*_{T_{\mathrm{f}}|k}}$ and $\tilde{\bm{y}}_{\mathrm{f},k+1} = \col{\bm{y}^*_{1|k},\dots,\bm{y}^*_{T_{\mathrm{f}}-1|k},\bm{E}_y \hat{\bm{A}}_{\mathrm{cl}}\bm{\xi}^*_{T_{\mathrm{f}}|k}}$. By design of the constraints~\eqref{eq:ocp_outputcons}, \eqref{eq:ocp_inputcons}, and \eqref{eq:ocp_termcons}, via Asm.~\ref{asm:termIngredients}, OCP~\eqref{eq:ocp} is feasible for these candidates.

Satisfaction of the chance constraints~\eqref{eq:constraints}, conditioned on the most recent feasible measured state $\bm{\xi}_{\tilde{k}}$,
follows from arguments presented in \cite[Thm.~2]{hewing2020recursively} and \cite[Cor.~8]{knaup2024recursively}: Let $\tilde{k}$ be the most recent time step for which OCP~\eqref{eq:ocp} was feasible using the initialization $\bm{\xi}_{0|\tilde{k}} = \bm{\xi}_{\tilde{k}}$.
Thus, the extended-state error at time $\tilde{k}$ is $\bm{e}^{\xi}_{\tilde{k}} = \bm{\xi}_{\tilde{k}} - \bm{\xi}_{0|\tilde{k}} = \bm{0}$. Now, let $\ol{k}\in\mathbb{N}$ be the number of consecutive time steps where the measured states $\bm{\xi}_{\tilde{k}+1},\dots,\bm{\xi}_{\tilde{k}+\ol{k}}$ are infeasible for OCP~\eqref{eq:ocp}, i.e., the indirect feedback initialization $\bm{\xi}_{0|l+1} = \bm{\xi}_{1|l}$ is used for $l\in\mathbb{N}_{\tilde{k}}^{\tilde{k}+\ol{k}-1}$. Furthermore, let $\bm{v}^*_{\tilde{k}},\dots,\bm{v}^*_{\tilde{k}+\ol{k}-1}$ be the corresponding sequence of closed-loop optimal controls. Then, the error terms satisfy $[\bm{G}_y]_i\bm{e}^y_{l} = [\bm{G}_y]_i\bm{e}^y_{l-{\tilde{k}}|\tilde{k}}$ and  $[\bm{G}_u]_j\bm{e}^u_{l} = [\bm{G}_u]_j\bm{e}^u_{l-{\tilde{k}}|\tilde{k}}$.
Since the corresponding predicted inputs and outputs satisfy the tightened constraints~\eqref{eq:ocp_inputcons} and \eqref{eq:ocp_outputcons} with the tightening parameters $\tilde{\eta}^u_l$ and $\tilde{\eta}^y_l$, the assertion follows from Lem.~\ref{lem:constight}.

\end{document}